\ifx\DCGVersion\undef
   \newcommand{\DCGVer}[1]{}
   \newcommand{\NotDCGVer}[1]{#1}
\else
   \newcommand{\DCGVer}[1]{#1}
   \newcommand{\NotDCGVer}[1]{}
\fi

\IfFileExists{/home/sariel/sariel/.login}%
{%
   \typeout{This is Sariel's computer!}
}{%
}%

\ifx\HsienRadical\undefined
\newcommand{\InHsienRadical}[1]{}%
\newcommand{\NotInHsienRadical}[1]{#1}%
\else
  \newcommand{\InHsienRadical}[1]{#1}%
  \newcommand{\NotInHsienRadical}[1]{}%
  
\fi

\NotDCGVer{%
   \InHsienRadical{ 
      \documentclass[11pt,twoside]{article} 
      \usepackage[margin=1in]{geometry}
      
   }
   \NotInHsienRadical{ 
      \documentclass[12pt]{article}
   }%
}

\DCGVer{%
   \documentclass[smallextended,envcountsame]{svjour3} %
}

\NotDCGVer{%
   \NotInHsienRadical{
      \usepackage[cm]{fullpage}%
   }%
}

\DCGVer{%
   \smartqed  % flush right qed marks, e.g. at end of proof
   \renewcommand{\qed}{\hfill$\square$}
}

\usepackage{graphicx}%
\usepackage{amsmath}%
\usepackage{xcolor}%
\usepackage{xspace}%
\usepackage{paralist}%
\usepackage{amssymb}
\usepackage{caption}%
\usepackage{picins}%
\usepackage{mleftright}%
\usepackage{mathtools}%
\usepackage{bm}%

\usepackage{amsmath}%
\usepackage{euscript}%

\NotDCGVer{%
   \usepackage{titlesec}%
   \titlelabel{\thetitle. }%

   \usepackage[amsmath,thmmarks]{ntheorem}
   \theoremseparator{.}%
}

\usepackage{hyperref}%
\NotDCGVer{%
   \hypersetup{%
      breaklinks,%
      colorlinks=true,%
      urlcolor=[rgb]{0.2,0.0,0.0},%
      linkcolor=[rgb]{0.5,0.0,0.0},%
      citecolor=[rgb]{0,0,0.445},%
      filecolor=[rgb]{0,0,0.4},
      anchorcolor=[rgb]={0.0,0.1,0.2}%
   }%
}

\NotDCGVer{%
   \numberwithin{figure}{section}%
   \numberwithin{table}{section}%
   \numberwithin{equation}{section}%

   \theoremstyle{plain}
   \newtheorem{theorem}{Theorem}[section]
   \newtheorem{lemma}[theorem]{Lemma}%
   \newtheorem{corollary}[theorem]{Corollary}
}

\NotDCGVer{%
   \newtheorem{observation}[theorem]{Observation} %
}
\NotDCGVer{%
   \renewtheorem*{observation*}[theorem]{Observation} %

   {\theorembodyfont{\rm} }
}
\NotDCGVer{%
   \newcommand{\myqedsymbol}{\rule{2mm}{2mm}}
}

\newcommand{\naive}{na\"ive\xspace}%
 
\newcommand{\si}[1]{#1}

\newcommand{\query}{\mathsf{x}}

\newcommand{\pnt}{\mathsf{p}}
\newcommand{\pntA}{\mathsf{q}}

\newcommand{\BclX}[1]{\mathsf{B}\pth{#1}}
\newcommand{\dist}[2]{\left\| #1 - #2 \right\| }

\newcommand{\PntSet}{\mathsf{P}}

\newcommand{\proxySet}[3]{{\Phi}_{#1}(#2, #3)}

\newcommand{\candid}{\mathsf{C}}

\newcommand{\HLinkShort}[2]{\hyperref[#2]{#1\ref*{#2}}}
\newcommand{\HLink}[2]{\hyperref[#2]{#1~\ref*{#2}}}
\newcommand{\HLinkSuffix}[3]{\hyperref[#2]{#1\ref*{#2}{#3}}}
\newcommand{\HLinkPage}[2]{\hyperref[#2]{#1~\ref*{#2}%
      $_\text{p\pageref{#2}}$}}

\newcommand{\figlab}[1]{\label{fig:#1}}
\newcommand{\figref}[1]{\HLink{Figure}{fig:#1}}

\newcommand{\apndlab}[1]{\label{apnd:#1}}
\newcommand{\apndref}[1]{\HLink{Appendix}{apnd:#1}}

\providecommand{\deflab}[1]{\label{def:#1}}
\newcommand{\defref}[1]{\HLink{Definition}{def:#1}}

\newcommand{\lemlab}[1]{\label{lemma:#1}}
\newcommand{\lemref}[1]{\HLink{Lemma}{lemma:#1}}%

\newcommand{\seclab}[1]{\label{sec:#1}}
\newcommand{\secref}[1]{\HLink{Section}{sec:#1}}

\newcommand{\corlab}[1]{\label{cor:#1}}
\newcommand{\corref}[1]{\HLink{Corollary}{cor:#1}}%

\newcommand{\thmlab}[1]{{\label{theo:#1}}}
\newcommand{\thmref}[1]{\HLink{Theorem}{theo:#1}}

\providecommand{\eqlab}[1]{}%
\renewcommand{\eqlab}[1]{\label{equation:#1}}
\newcommand{\Eqref}[1]{\HLinkSuffix{Eq.\ (}{equation:#1}{)}}

\newcommand{\etal}{\textit{et~al.}\xspace}

\newcommand{\atmostK}[3]{\mathsf{N}_{#2}\pth{#3, #1}}

\newcommand{\Set}[2]{\left\{ #1 \,\middle|\, #2 \right\}}

\newcommand{\brc}[1]{\left\{ {#1} \right\}}
\newcommand{\cardin}[1]{\left| {#1} \right|}

\newcommand{\pth}[1]{\mleft({#1}\mright)}

\newcommand{\pbrcx}[1]{\left[ {#1} \right]}

\newcommand{\ExChar}{\mathbf{E}}

\InHsienRadical{
\newcommand{\Ex}[2][\!]{{\ExChar}#1\pbrcx{#2}}

\newcommand{\ExOver}[2]{{{\ExChar}_{#1}}\!\pbrcx{#2}}
\newcommand{\ExCond}[3][\!]{{\ExChar} #1%
   \left[ #2 \;\middle|\; #3 \right]}
\newcommand{\ExOverCond}[3]{%
   {\ExChar}_{#1}\!\left[ #2 \;\middle\vert\; #3 \right]}
}
\NotInHsienRadical{
\newcommand{\Ex}[2][\!]{\mathop{\ExChar}#1\pbrcx{#2}}

\newcommand{\ExOver}[2]{{{\ExChar}_{#1}}\!\pbrcx{#2}}
\newcommand{\ExCond}[3][\!]{\mathop{\ExChar} #1%
   \left[ #2 \;\middle|\; #3 \right]}
\newcommand{\ExOverCond}[3]{%
   {\ExChar}_{#1}\!%
   \left[ #2 \;\middle\vert\; #3 \right]}
}

\newcommand{\Prob}[1]{{\mathbf{Pr}}\!\pbrcx{#1}}
\newcommand{\sep}[1]{\,\big.\left|\, {#1} \right.}
\newcommand{\permut}[1]{\left\langle {#1} \right\rangle}%
\InHsienRadical{
\newcommand{\PermutSet}{\smash{\mathrm{\hat\Pi}}}%
\newcommand{\Permut}{\tau}%
\newcommand{\PermutE}{\tau}%
}
\NotInHsienRadical{
\newcommand{\PermutSet}{\Xi}%
\newcommand{\Permut}{{\tau}}%
\newcommand{\PermutE}{\tau}%
}

\newcommand{\SeqSetY}[2]{\pth{#1}_{#2}}%
\newcommand{\object}{\tau}

\newcommand{\SiteSet}{\mathsf{S}}

\newcommand{\element}{\mathsf{s}}
\newcommand{\site}{\mathsf{s}}

\DeclareFontFamily{OT1}{pzc}{}
\DeclareFontShape{OT1}{pzc}{m}{it}{<-> s * [1.10] pzcmi7t}{}
\DeclareMathAlphabet{\mathpzc}{OT1}{pzc}{m}{it}

\InHsienRadical{
	\newcommand{\stair}{\mathsf{St}}
}
\NotInHsienRadical{
\newcommand{\stair}{\!\raisebox{-0.4ex}{%
      {{\includegraphics[height=2.3ex]{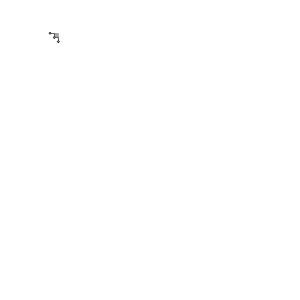}}}%
   }%
   \hspace*{-1.6pt}%
}}

\newcommand{\volX}[1]{\mathrm{vol}\pth{#1}}

\newcommand{\VolUArea}[2]{F_{#1}\pth{#2}}

\newcommand{\ElemSet}{\mathsf{S}}
\newcommand{\RSample}{\mathsf{R}}

\newcommand{\stairX}[1]{\stair\pth{#1}}%
\newcommand{\OPntSet}{\mathsf{Q}}%

\newcommand{\Property}{\mathcal{P}} %
\newcommand{\PSetX}[1]{\Property\pth{#1}}%

\newcommand{\valsY}{\mathsf{Y}}
\newcommand{\dInt}{\mathrm{d}}

\newcommand{\PntSetA}{\mathsf{X}}

\newcommand{\DefSet}[1]{D\pth{#1}}
\newcommand{\KillSet}[1]{K\pth{#1}}

\renewcommand{\th}{\si{th}\xspace}

\newcommand{\eps}{\varepsilon}%
\renewcommand{\Re}{\mathbb{R}}

\definecolor{blue25}{rgb}{0,0,0.55}%
\newcommand{\emphic}[2]{%
   \textcolor{blue25}{%
      \textbf{\emph{#1}}}%
   \index{#2}}

\InHsienRadical{
	\definecolor{DarkRed}{rgb}{0.50,0.00,0.00}
	\def\emphi#1{\textcolor{DarkRed}{{\emph{#1}}}}
	\pdfstringdefDisableCommands{\let\boldmath\relax} % allow \boldmath in section titles
}
\NotInHsienRadical{
	\newcommand{\emphi}[1]{\emphic{#1}{#1}}
}

\newcommand{\atgen}{\symbol{'100}}%

   \newcommand{\SarielThanks}[1]{%
      \thanks{%
         Department of Computer Science, %
         University of Illinois; %
         201 N.\ Goodwin Avenue, %
         Urbana, IL 61801, USA; %
         {\tt \si{sariel}\atgen{}\si{uiuc.edu}}; %
         {\tt \url{sarielhp.org}}.%
         #1%
      }%
   }

\newcommand{\HsienThanks}[1]{%
   \thanks{%
      Department of Computer Science, %
      University of Illinois; %
      {\tt \si{hchang17}\atgen{}\si{illinois}.\si{edu}}; %
      {\tt \url{illinois.edu/\string~hchang17}}.%
      #1%
   }%
}

\newcommand{\BenThanks}[1]{%
   \thanks{%
      Department of Computer Science, %
      University of Texas at Dallas; %
      800 W. Campbell Rd., MS EC-31,
      Richardson TX 75080, USA; %
      {\tt \si{benjamin.raichel}\atgen{}\si{utdallas}.\si{edu}}; %
      {\tt \url{\si{utdallas.edu/\string~\si{bar150630}}}}%
      . %
      #1%
   }
}

\newcommand{\ObjX}[1]{\mathcal{T}\pth{#1}}

\newcommand{\remove}[1]{}

\newcommand{\attribB}{\beta}

\newcommand{\Distribution}{\mathcal{D}}

\newcommand{\polylog}{\operatorname{polylog}}

\NotInHsienRadical{
   \newcommand{\Event}{E}
   \newcommand{\EventA}{\mathcal{E}}
   \newcommand{\EventB}{\mathcal{F}}
}
\InHsienRadical{
   \newcommand{\Event}{E}
   \newcommand{\EventA}{E}
   \newcommand{\EventB}{F}
}

\newcommand{\env}{\mathrm{env}}%
\newcommand{\kdistSet}[3]{\mathsf{d}_#1\pth{#2, #3}}%

\newcommand{\plane}{h}

\newcommand{\HPlanes}{\mathsf{H}}%

\newcommand{\bclX}[1]{\mathsf{b}_{#1}}%

\newcommand{\Verts}{\mathsf{V}}%

\newcommand{\VLvl}[2]{\Verts_{#1}\pth{#2}}

\newcommand{\vertex}{v}%
\newcommand{\vertexA}{u}%
\newcommand{\Edges}{\mathsf{E}}%
\newcommand{\edge}{e}%
\newcommand{\planeLevel}[2]{\Edges_{#1}\pth{#2}}

\DefineNamedColor{named}{RedViolet} {cmyk}{0.07,0.90,0,0.34}

\newcommand{\storage}{space\xspace}

\InHsienRadical{
	\newcommand{\newK}{k'}
}

\NotInHsienRadical{
	\newcommand{\newK}{\Bbbk}
}

\providecommand{\Property}{\Pi}
\providecommand{\SetX}{\mathsf{X}}
\providecommand{\kdistSet}[3]{\mathsf{d}_{#1}(#2, #3)}%
\providecommand{\Owhp}{ {O_{\text{whp}}} }
\providecommand{\pv}{\mathrm{pv}}
\providecommand{\VolUArea}[2]{F_{#1}\!\Paren{#2}}
\providecommand{\env}{\mathrm{env}}
\providecommand{\envSet}[3]{{\env}_{#1}(#2, #3)}%
\providecommand{\Edges}{\mathsf{E}}
\providecommand{\planeLevel}[2]{\Edges_{#1}({#2})}
\providecommand{\DefSet}{\mathrm{Def}}
\providecommand{\KillSet}{\mathrm{Stop}}
\providecommand{\Event}{E}
\NotDCGVer{%
	\NotInHsienRadical{%
	   \def\SarielDefStyle{TRUE}
	}

   \ifx\SarielDefStyle\undefined
      \theoremstyle{remark}
      \theoremseparator{.}%
   
      {\theorembodyfont{\rm} \newtheorem{remark}[theorem]{Remark}}
	{\theorembodyfont{\rm} \renewtheorem*{remark*}{Remark}}
      {\theorembodyfont{\rm} \newtheorem{defn}[theorem]{Definition}}
     {\theorembodyfont{\rm} \newtheorem{example}[theorem]{Example}}

   \else 
      \theoremstyle{plain}%
      \theoremheaderfont{\sf}
      \theorembodyfont{\upshape}%
      \theoremseparator{.}%
      
      \newtheorem{defn}[theorem]{Definition}
      
      \newtheorem*{remark:unnumbered}{Remark}%
      \newtheorem{remark}[theorem]{Remark}%
      \renewtheorem*{remark*}{Remark} %

   \fi
}

\DCGVer{%
   \spnewtheorem{defn}[theorem]{Definition}{\itshape}{\rmfamily}
   \spnewtheorem{observation}[theorem]{Observation}{\itshape}{\rmfamily}
}

\NotDCGVer{%
   \theoremheaderfont{\em}%
   \theorembodyfont{\upshape}%
   \theoremstyle{nonumberplain}
   \theoremseparator{}
   \theoremsymbol{\myqedsymbol}
   \newtheorem{proof}{Proof:}
}
\newcommand{\myparagraph}[1]{\paragraph{#1}}

\InHsienRadical{
\newcommand{\sfx}{\sigma}%
}

\NotInHsienRadical{
\newcommand{\sfx}{\mathbf{\sigma}}%
}

\newcommand{\SuffixY}[2]{{#1}_{#2}^n}
\InHsienRadical{% 
   \usepackage[mathlines, pagewise]{lineno}
	\linenumbers
	
	\setlength\linenumbersep{2em}

	\newcommand*\patchAmsMathEnvironmentForLineno[1]{%
	 \expandafter\let\csname old#1\expandafter\endcsname\csname #1\endcsname
	 \expandafter\let\csname oldend#1\expandafter\endcsname\csname end#1\endcsname
	 \renewenvironment{#1}%
	 {\linenomath\csname old#1\endcsname}%
	 {\csname oldend#1\endcsname\endlinenomath}}% 
	\newcommand*\patchBothAmsMathEnvironmentsForLineno[1]{%
	 \patchAmsMathEnvironmentForLineno{#1}%
	 \patchAmsMathEnvironmentForLineno{#1*}}%
	\AtBeginDocument{%
	\patchBothAmsMathEnvironmentsForLineno{equation}%
	\patchBothAmsMathEnvironmentsForLineno{align}%
	\patchBothAmsMathEnvironmentsForLineno{flalign}%
	\patchBothAmsMathEnvironmentsForLineno{alignat}%
	\patchBothAmsMathEnvironmentsForLineno{gather}%
	\patchBothAmsMathEnvironmentsForLineno{multline}%
	}
}

\DCGVer{%
   \DefineNamedColor{named}{RedViolet} {rgb}{0, 0, 0}%
   \hypersetup{%
      breaklinks,%
      colorlinks=true,%
      urlcolor=[rgb]{0.0,0.0,0.0},%
      linkcolor=[rgb]{0.0,0.0,0.0},%
      citecolor=[rgb]{0,0,0},%
      filecolor=[rgb]{0,0,0},
      anchorcolor=[rgb]={0.0,0.0,0.0}%
   }%
}

\begin{document}

\title{From Proximity to Utility: %
   \NotDCGVer{\\}%
   A Voronoi Partition of Pareto Optima%
   \NotDCGVer{%
      \footnote{%
         Work on this paper was partially supported by NSF AF award
         CCF-1421231, and % Started June 2014
         CCF-1217462.  % Started June 2012
         A preliminary version of the paper appeared in the 31st
         International Symposium on Computational Geometry (SoCG 2015)
         \cite{chr-fpuvp-15}.  The paper is also available on the
         arXiv \cite{chr-fpuvp-14}. %
      }%
   }%
}

\author{%
   Hsien-Chih Chang%
   \NotDCGVer{%
      \HsienThanks{}%
   }
   \and%
   Sariel Har-Peled%
   \NotDCGVer{%
      \SarielThanks{}%
   }
   \and%
   Benjamin Raichel%
   \NotDCGVer{%
      \BenThanks{}%
   } }

\DCGVer{%
   \institute{%
      H.-C.\ Chang%
      \and S.\ Har-Peled%
      \at%
      Department of Computer Science; University of Illinois; 201 N.\
      Goodwin Avenue; Urbana, IL, 61801, USA.%\\
      \and%
      B.\ Raichel%
      \at Department of Computer Science; %
      University of Texas at Dallas; %
      800 W.\ Campbell Rd., MS EC-31; Richardson, TX 75080; USA. %
   }%
}

\date{\today}

\remove{%
   \author[1]{Hsien-Chih Chang}%
   \author[1]{Sariel Har-Peled}%
   \author[1]{Benjamin Raichel}%
   \affil[1]{%
      Department of Computer Science, %
      University of Illinois\\ %
      201 N. Goodwin Avenue, %
      Urbana, IL, 61801, USA.\\ %
      \texttt{{hchang17}@\si{illinois}.\si{edu}}, %
      \texttt{{sariel}@\si{illinois}.\si{edu}}, %
      \texttt{\si{raichel}2@\si{illinois}.\si{edu}} %
   }%
}

\remove{%
   \authorrunning{H.-C.\ Chang, S.\ Har-Peled and B.\ Raichel}
   
   \Copyright{Hsien-Chih Chang, Sariel Har-Peled and Benjamin
      Raichel}%

   %
   % mandatory: Please choose ACM 1998 classifications from
   % http://www.acm.org/about/class/ccs98-html . E.g., cite
   % as "F.1.1 Models of Computation".

   \serieslogo{}%please provide file name (without suffix)
   \volumeinfo%(easy chair interface)
   {}% editors
   {2}% number of editors: 1, 2, ....
   {Conference title on which this volume is based on}% event
   {1}% volume
   {1}% issue
   {1}% starting page number
   \EventShortName{}
   \DOI{10.4230/LI{}PI{}cs.xxx.yyy.p}% to be completed by the volume editor
}

\maketitle

\begin{abstract}
    We present an extension of Voronoi diagrams where when considering
    which site a client is going to use, in addition to the site
    distances, other site attributes are also considered (for example, prices
    or weights).  A cell in this diagram is then the locus of all
    clients that consider the same set of sites to be relevant.  In
    particular, the precise site a client might use from this
    candidate set depends on parameters that might change between
    usages, and the candidate set lists all of the relevant sites. The
    resulting diagram is significantly more expressive than Voronoi
    diagrams, but naturally has the drawback that its complexity, even
    in the plane, might be quite high.  Nevertheless, we show that if
    the attributes of the sites are drawn from the same distribution
    (note that the locations are fixed), then the expected complexity
    of the candidate diagram is near linear.

    To this end, we derive several new technical results, which are 
    of independent interest.  In particular, we provide a high-probability, 
    asymptotically optimal bound on the number of Pareto optima points 
    in a point set uniformly sampled from the $d$-dimensional 
    hypercube.  To do so we revisit the classical backward analysis 
    technique, both simplifying and improving relevant results in order 
    to achieve the high-probability bounds.
\end{abstract}

%%%%%%%%%%%%%%%%%%%%%%%%%%%%%%%%%%%%%%%%%%%%%%%%%%%%%%%%%%%%%%%%%%%%%% 
%%%%%%%%%%%%%%%%%%%%%%%%%%%%%%%%%%%%%%%%%%%%%%%%%%%%%%%%%%%%%%%%%%%%%% 

%%%%%%%%%%%%%%%%%%%%%%%%%%%%%%%%%%%%%%%%%%%%%%%%%%%%%%%%%%%%%%%%%%%%%% 
%%%%%%%%%%%%%%%%%%%%%%%%%%%%%%%%%%%%%%%%%%%%%%%%%%%%%%%%%%%%%%%%%%%%%% 

\section{Introduction}

\myparagraph{Informal
   description of the candidate diagram.} %
Suppose you open your refrigerator one day to discover it is time to
go grocery shopping.%
\footnote{Unless you are feeling adventurous enough that day to eat
   the frozen mystery food stuck to the back of the freezer, which we
   strongly discourage you from doing.}  Which store you go to will be
determined by a number of different factors.  For example, what items
you are buying, and do you want the cheapest price or highest quality,
and how much time you have for this chore.  Naturally the distance to
the store will also be a factor.  On different days which store is the
best to go to will differ based on that day's preferences.  However,
there are certain stores you will never shop at.  These are stores
which are worse in every way than some other store (further,
more expensive, lower quality, etc).  Therefore, the stores that are
relevant and in the \emph{candidate set} are those that are not
strictly worse in every way than some other store.  Thus, every point
in the plane is mapped to a set of stores that a client at that
location might use.  The \emph{candidate diagram} is the partition of
the plane into regions, where each candidate set is the same for all
points in the same region.  Naturally, if your only consideration is
distance, then this is the (classical) Voronoi diagram of the sites.
However, here deciding which shop to use is an instance of
multi-objective optimization --- there are multiple, potentially
competing, objectives to be optimized, and the decision might change
as the weighting and influence of these objectives mutate over time
(in particular, you might decide to do your shopping in different
stores for different products).  The concept of relevant stores
discussed above is often referred as the \emph{Pareto optima}.

\myparagraph{Pareto optima in welfare economics.} %
Pareto efficiency, named after Vilfredo Pareto, is a core concept in
economic theory and more specifically in welfare economics.  Here each
point in $\Re^d$ represents the corresponding utilities of $d$ players
for a particular allocation of finite resources.  A point is said to
be \emph{Pareto optimal} if there is no other allocation which
increases the utility of any individual without decreasing the utility
of another.  The \emph{First Fundamental Theorem of Welfare Economics}
states that any competitive equilibrium (when supply equals demand)
is Pareto optimal.  The origins of this theorem date back to 1776 with
Adam Smith's famous (and controversial) work, ``The Wealth of
Nations,'' but was not formally \emph{proven} until the 20\th century
by Lerner, Lange, and Arrow (see \cite{a-we-08}).  Naturally such
proofs rely on simplifying (and potentially unrealistic) assumptions
such as perfect knowledge, or absence of externalities.
The \emph{Second Fundamental Theorem of Welfare Economics} states that
any Pareto optimum is achievable through lump-sum transfers (that is,
taxation and redistribution).  In other words each Pareto optima is a
``best solution'' under some set of societal preferences, and is
achievable through redistribution in one form or another (see
\cite{a-we-08} for a more in depth discussion).

\myparagraph{Pareto optima in computer science.} %
In computational geometry such Pareto optima points relate to the
\emph{orthogonal convex hull} \cite{osw-odcrc-84}, which in turn
relates to the well known convex hull (the input points that lie on
the orthogonal convex hull is a super set of those which lie on the
convex hull).  Pareto optima are also of importance to the database
community \cite{bks-so-01,htc-tpkds-13}, in which context such points
are called \emph{maximal} or \emph{skyline points}.  Such points are
of interest as they can be seen as the relevant subset of the
(potentially much larger) result of a relational database query.  The
standard example is querying a database of hotels for the cheapest and
closest hotel, where naturally hotels which are farther and more
expensive than an alternative hotel are not relevant results.  There
is a significant amount of work on computing these points, see Kung
\etal~\cite{klp-fmsv-75}. More recently, Godfrey
\etal~\cite{gsg-aamvc-07} compared various approaches for the
computation of these points (from a databases perspective), and also
introduced their own new external algorithm.%
\footnote{There is of course a lot of other work on Pareto optimal
   points, from connections to Nash equilibrium to scheduling. We
   resisted the temptation of including many such references which are
   not directly related to our paper.}%
% For other relevant work, see INSERT MORE REFS HERE.
% \myparagraph{Modeling uncertainty.} %

\paragraph{Modeling uncertainty.} %
Recently, there is a growing interest in modeling uncertainty in data.
As real data is acquired via physical measurements, noise and errors
are introduced.  This can be addressed by treating the data as coming
from a distribution (e.g., a point location might be interpreted as a
center of a Gaussian), and computing desired classical quantities
adapted for such settings.  Thus, a nearest-neighbor query becomes a
probabilistic question --- what is the expected distance to the
nearest-neighbor?  What is the most likely point to be the
nearest-neighbor?
% Etc.
(See \cite{aahpy-nnsuu-13} and references therein.)

This in turn gives rise to the question of what is the expected
complexity of geometric structures defined over such data.  The case
where the data is a set of points, and the locations of the points are
chosen randomly was thoroughly investigated (see \cite{sw-ig-93,%
   ww-sg-93,hr-ecrwv-14} and references therein).  The problem, when
the locations are fixed but the weights associated with the points are
chosen randomly, is relatively new.  Agarwal
\etal~\cite{ahks-urmsn-14} showed that for a set of disjoint segments
in the plane, if they are being expanded randomly, then the expected
complexity of the union is near linear.  This result is somewhat
surprising as in the worst case the complexity of such a union is
quadratic.

Here we are interested in bounding the expected complexity of weighted
generalizations of Voronoi diagrams, where the
weights (not the site locations) are randomly sampled.  Note that the
result of Agarwal \etal~\cite{ahks-urmsn-14} can be interpreted as
bounding the expected complexity of level sets of the multiplicative
weighted Voronoi diagram (of segments).  On the other hand, we
want to bound the entire lower envelope (which implies the same bound
on any level set).  For the special case of multiplicative weighted
Voronoi diagrams, a near-linear expected complexity bound was provided
by Har-Peled and Raichel \cite{hr-ecrwv-14}.  In this work we consider
a much more general class of weighted diagrams which allow multiple
weights and non-linear distance functions.

\subsection{Our contributions}

\myparagraph{Conceptual contribution.} %
We formally define the \emph{candidate diagram} in
\secref{beer:diagram:def} --- a new geometric structure that combines
proximity information with utility.  For every point $\query$ in the
plane, the diagram associates a \emph{candidate set} $\candid(\query)$
of sites that are relevant to $\query$. That is, all the sites that
are Pareto optima for $\query$.  Putting it differently, a site is not
in $\candid(\query)$ if it is further away from \emph{and} worse in
all parameters than some other site.  Significantly, unlike the
traditional Voronoi diagram, the candidate diagram allows the user to
change their distance function, as long as the function respects the
domination relationship.
This diagram is a significant extension of the Voronoi diagram, and
includes other extensions of Voronoi diagrams as special subcases,
like multiplicative weighted Voronoi diagrams.  Not surprisingly, the
worst case complexity of this diagram can be quite high.

\myparagraph{Technical contribution.} %
We consider the case where each site chooses its $j$\th attribute from
some distribution $\Distribution_j$ independently for each $j$.  We
show that the candidate diagram in expectation has near-linear
complexity, and that, with high probability, the candidate set has
poly-logarithmic size for any point in the plane. In the process we
derive several results which are interesting in their own right.
\smallskip%
\InHsienRadical{\vspace{4pt}\begin{compactenum}[(a)]}
\NotInHsienRadical{\begin{compactenum}[(A)]}
    \item 
    % \InHsienRadical{\textit{Low complexity of the minima for random
    % points in athe hypercube.}}
    \NotInHsienRadical{\textbf{Low complexity of the minima for random points in
       the hypercube.}}  
    We prove that if $n$ points are sampled from a
    fixed distribution (see \secref{sample:model} for assumptions on
    the distribution) over the $d$-dimensional hypercube then, with
    \emph{polynomially} small error probability, the number of Pareto
    optima points is $O(\log^{d-1} n)$, which is within a constant
    factor of the expectation (see \lemref{staircase:w:h:p}).
    Previously, this result was only known in a weaker form that is
    insufficient to imply our other results.  Specifically, Bentley
    \etal~\cite{bkst-anmsv-78} first derived the asymptotically tight
    bound on the expected number of Pareto optima points. Bai
    \etal~\cite{bdht-mh-05} proved that after normalization the
    cumulative distribution function of the number of Pareto optima is
    normal, up to an additive error of $O\pth{1/\!\polylog n}$.  (See
    \cite{br-ovrp-10,br-pp-10} as well.)  In particular, their results
    (which are quite nice and mathematically involved) only imply our
    statement with poly-logarithmically small error probability.
    To the best of our knowledge this result is new --- we emphasize,
    however, that for our purposes a weaker bound of $O(\log^d n)$ is
    sufficient, and such a result follows readily from the $\eps$-net
    theorem \cite{hw-ensrq-87} (naturally, this would add a
    logarithmic factor to later results).

    \smallskip%
    \item 
    \InHsienRadical{\textit{Backward analysis with high probability.}}
    \NotInHsienRadical{\textbf{Backward analysis with high probability.}} %
    To get this result, we prove a lemma providing high-probability
    bounds when applying backwards analysis \cite{s-barga-93} (see
    \lemref{property:e}). Such tail estimates are known in the context
    of randomized incremental algorithms \cite{cms-frric-93,%
       bcko-cgaa-08}, but our proof is arguably more direct and
    cleaner, and should be applicable to more cases.  (See
    \secref{minima:whp}).
    
    \smallskip%
    \item \InHsienRadical{\textit{Overlay of the $k$\th order Voronoi
          cells in randomized incremental construction.}}
    \NotInHsienRadical{\textbf{Overlay of the $k$\th order Voronoi
          cells in randomized incremental construction.}}
    We prove that the overlay of cells during a randomized incremental
    construction of the $k$\th order Voronoi diagram is of complexity
    $O\pth{k^4 n \log n }$ (see \lemref{proxy:complexity}).
    
    \smallskip%
    \item \InHsienRadical{\textit{Complexity of the candidate
          diagram.}}  \NotInHsienRadical{\textbf{Complexity of the
          candidate diagram.}} %
    Combining the above results carefully yields a near-linear upper
    bound on the complexity of the candidate diagram (see
    \thmref{candid:complexity}).
    % A high level sketch of this process is described below.
\end{compactenum}

% \myparagraph{Relation to previous work.} %
%%
% As mentioned above, Har-Peled and Raichel \cite{hr-ecrwv-14}
% provided a near-linear bound on the expected complexity of
% multiplicative Voronoi diagrams.  We emphasize several key
% distinctions between this previous work and the current paper.
% First, candidate diagrams (in a weaker form) were only used
% implicitly in \cite{hr-ecrwv-14} to bound the complexity of the
% multiplicative diagram.  Therefore the current paper is the first to
% consider the candidate diagram as an independent object of study.
% Equally important is the generalization to weight vectors as opposed
% to a single weight.  Indeed, when there is only a single weight
% involved the candidate diagram per site is less interesting.
% Moreover, in \cite{hr-ecrwv-14} the distance function considered was
% always the product of the scalar weight of the site and its
% distance.  Here we allow for the weight function to be any
% (coordinate-wise) non-decreasing function of the weight vectors.  In
% particular, this weight function can be non-linear and is allowed to
% vary (while the diagram remains fixed).

\myparagraph{Outline.} %
In \secref{candidDiagram} we formally define our problem and introduce
some tools that will be used later on.  Specifically, after some
required preliminaries, we formally introduce the candidate diagram in
\secref{beer:diagram:def}. The sampling model being used is described in
detail in \secref{sample:model}.  

Backward analysis with high probability is discussed in
\secref{minima:whp}, including \corref{property} which is a sufficient
statement for the purposes of this paper.  In
\secref{backward:analysis:b} we make a short detour and provide a
detailed proof of the high-probability backward analysis statement.

To bound the complexity of the candidate diagram (both the size
of the planar partition and the total size of the associated candidate
sets), in \secref{proxy:set}, the notion of \emph{proxy
   set} is introduced. Defined formally in \secref{proxy:set:def}, it is
(informally) an enlarged candidate set.  \secref{proxy:set:size} bounds
the size of the proxy set using backward analysis, both in expectation
and with high probability, and \secref{proxy:contain:beer} shows that
mucking around with the proxy set is useful, by proving that the proxy
set contains the candidate set, for any point in the plane.

In \secref{k:th:order}, it is shown that the diagram induced by the
proxy sets can be interpreted as the arrangement formed by the overlay
of cells during the randomized incremental construction of the $k$\th
order Voronoi diagram.  
To this end, \secref{k:th:order:prelim} defines the \emph{$k$\th order 
Voronoi diagram}, interpret as arrangement of planes, and states some basic 
properties of these entities.  For our purposes, we need to bound the 
size of the conflict lists encountered during the randomized incremental 
construction, and this is done in \secref{below:conflict:size} using the 
Clarkson-Shor technique. In \secref{enviroments:overlays} the 
\emph{$k$ environment} of a site is defined, and we related such a notion 
to the $k$\th order Voronoi diagram. Next, in \secref{all:together} we 
bound the expected complexity of the proxy diagram.
%
%To this end, \secref{k:th:order:prelim}
%defines the $k$\th order Voronoi diagram, as well as the \emph{$k$\!
%   environment} of a site, and states some basic properties of these
%entities.  For our purposes, we need to bound the size of the conflict
%lists encountered during the randomized incremental construction, and
%this is done in \secref{below:conflict:size} using the Clarkson-Shor
%technique. Next, in \secref{all:together} we bound the expected
%complexity of the proxy diagram.

We bound the expected size of the candidate set for any point in the
plane in \secref{size:candidate:set}.  First, in \secref{stair}, we
analyze the number of staircase points in randomly sampled point sets
from the hypercube, and we use this bound, in \secref{size:beer:set},
to bound the size of the candidate set.

Finally, in \secref{pf:candid:complexity}, we put everything together and
prove our main result, showing the desired bound on the complexity of
the candidate diagram.
%
%We fill in the missing details for the results of
%\secref{backward:analysis:a}, proving a high-probability bound for
%backward analysis in \RefSecMinimaWHP.
%\InNotProcVer{%
%   We delegated these details to the appendix, since these results are
%   somewhat orthogonal to the rest of the paper.}

%%%%%%%%%%%%%%%%%%%%%%%%%%%%%%%%%%%%%%%%%%%%%%%%%%%%%%%%%%%%%%% 
%%%%%%%%%%%%%%%%%%%%%%%%%%%%%%%%%%%%%%%%%%%%%%%%%%%%%%%%%%%%%%% 

\section{Problem definition and preliminaries}
\seclab{candidDiagram}

% \subsection{Preliminaries}
% \seclab{prelim}

Throughout, we assume the reader is familiar with standard
computational geometry terms, such as arrangements \cite{sa-dsstg-95},
vertical-decomposition \cite{bcko-cgaa-08}, etc.  In the same vein, we
assume that the variable $d$, the \emph{dimension}, is a small
constant and the big-$O$ notation hides constants that are potentially
exponential (or worse) in $d$.

A quantity is bounded by $O\pth{f}$ \emphi{with high probability} with
respect to $n$, if for any constant $\gamma > 0$, there is another
constant $c$ depending on $\gamma$ such that the quantity is at most
$c \cdot f$ with probability at least $1-n^{-\gamma}$.  In other words,
the bound holds for any polynomially small error with the expense of a
multiplicative constant factor on the size of the bound.  When there's
no danger of confusion, we sometimes write {$\Owhp(f)$}
for short.
\begin{defn}
    \deflab{dominates}%
    Consider two points $\pnt = (\pnt_1, \ldots, \pnt_d)$ and
    $\pntA = (\pntA_1, \ldots, \pntA_d)$ in $\Re^d$.  The point $\pnt$
    \emphi{dominates} $\pntA$ (denoted by $\pnt \preceq \pntA$) if
    $\pnt_i \leq \pntA_i$, for all $i$.
\end{defn}
Given a point set $\PntSet \subseteq \Re^d$, there are several terms
for the subset of $\PntSet$ that is not dominated, as discussed above,
such as \emph{Pareto optima} or \emph{minima}.  Here, we use the
following term.
\begin{defn}
    \deflab{stair}%
    For a point set $\PntSet \subseteq \Re^d$, a point
    $\pnt \in \PntSet$ is a \emphi{staircase point} of $\PntSet$ if no
    other point of $\PntSet$ dominates it.  The set of all such
    points, denoted by $\stairX{\PntSet}$, is the \emphi{staircase} of
    $\PntSet$.
\end{defn}
Observe that for a nonempty finite point set $\PntSet$, the staircase
$\stairX{\PntSet}$ is never empty.

\subsection{Formal definition of the candidate diagram}
\seclab{beer:diagram:def}

Let $\SiteSet = \brc{\site_1, \ldots, \site_n}$ be a set of $n$
distinct \emphi{sites} in the plane.  For each site $\site$ in
$\SiteSet$, there is an associated list
$\attribB = \pth{b_1, \ldots, b_d}$ of $d$ real-valued attributes,
each in the interval $[0,1]$.  When viewed as a point in the unit
hypercube $[0,1]^d$, this list of attributes is the \emphi{parametric
   point} of the site $\site_i$.  Specifically, a \emph{site} is a point in
the plane encoding a facility location, while the term \emph{point} is
used to refer to the (parametric) point encoding its attributes in
$\Re^d$.

\myparagraph{Preferences.} %
Fix a client location $\query$ in the plane.  For each site, there are
$d+1$ associated variables for the client to consider. Specifically,
the client distance to the site, and $d$ additional attributes (e.g.,
prices of $d$ different products) associated with the site.
Conceptually, the goal of the client is to ``pay'' as little as
possible by choosing the best site (e.g., minimize the overall cost of
buying these $d$ products together from a site, where the price of
traveling the distance to the site is also taken into account).

\begin{defn}
    \deflab{dominating:pref}%
    A client $\query$ has a \emphi{dominating preference} if for any
    two sites $\site$ and $\site'$ in the plane, with parametric points
    $\attribB$ and $\attribB'$ in $\Re^d$, respectively, the client would
    prefer the site $\site$ over $\site'$ if
    $\dist{\query}{\site} \leq \dist{\query}{\site'}$ and
    $\attribB \preceq \attribB'$ (that is, $\attribB$ dominates
    $\attribB'$).  We sometimes say site $\site$ \emph{dominates} site $\site'$.
    % In such a case, we would say that the client uses a
    % \emphi{dominating
    % distance function} in deciding which site to use.
\end{defn}

Note that a client having a dominating preference does not identify a
specific optimum site for the client, but rather a set of potential
optimum sites.  Specifically, given a client location $\query$ in the
plane, let its distance to the $i$\th site be
$\ell_i = \dist{\query}{\site_i}$.  The set of sites the client might
possibly use (assuming the client uses a dominating preference) are
the staircase points of the set
$\PntSet(\query) = \brc{ (\attribB_1, \ell_1), \ldots,
   (\attribB_n,\ell_n) }$ (that is, we are adding the distance to each
site as an additional attribute of the site --- this attribute depends
on the location of $\query$). The set of sites realizing the staircase
of $\PntSet(\query)$
%(that is, all the sites relevant to $\query$) 
is the \emphi{candidate set} $\candid(\query)$ of $\query$:
\begin{align}
    \candid(\query) = \brc{ \site_i \in \SiteSet %
       \sep{ \text{$(\attribB_i, \ell_i)$ is a staircase point of
             $\PntSet(\query)$ in $\Re^{d+1}$} }}.%
    \eqlab{candidate:set}%
\end{align}
The \emphi{candidate cell} of $\query$ is the set of all the points in
the plane that have the same candidate set associated with them.  That
is, $\brc{ \pnt \in \Re^2 \mid \candid(\pnt) = \candid(\query) }$.
The decomposition of the plane into these cells is the
\emphi{candidate diagram}.

Now, the client $\query$ has the candidate set $\candid(\query)$, and
it chooses some site (or potentially several sites) from
$\candid(\query)$ that it might want to use. Note that the client
might decide to use different sites for different acquisitions.
As an example, consider the case when each site $\site_i$ is attached
with attributes $\attribB_i = (b_{i,1}, b_{i,2})$.  If the client
$\query$ has the preference of choosing the site with smallest value
$b_{i,1} \ell_i$ among all the sites, then this preference is a
dominating preference, and therefore the client will choose one of the
sites from the candidate list $\candid(\query)$.  (Observe that the
preference function corresponds to the multiplicative Voronoi diagram
with respect to the first coordinate $b_{i,1}$.)  Similarly, if the
preference function is to choose the smallest value
$b_{i,1} \hspace{0.6pt} \ell_i^{2} + b_{i,2}$ among all the sites (which again
is a dominating preference), then this corresponds to a power diagram
of the sites.

% \subsection{Complexity of the diagram}
% \seclab{complexity:diagram}
\myparagraph{Complexity of the diagram.} %
The \emphi{complexity} of a planar arrangement %$\ArrC$
is the total number of edges, faces, and vertices.
% denoted by $\cardin{\ArrC}$.
A candidate diagram can be interpreted as a planar arrangement, and
its complexity is defined analogously.  The \emphi{\storage
   complexity} of the candidate diagram is the total amount of memory
needed to store the diagram explicitly, and is bounded by the
complexity of the candidate diagram together with the sum of the sizes
of candidate sets over all the faces in the arrangement of the diagram
(which is potentially larger by a factor of $n$, the number of
sites). Note, that the \storage complexity is a somewhat \naive upper
bound, as using persistent data-structures might significantly reduce
the space needed to store the candidate lists.

\begin{lemma}
   \lemlab{complexity}%
   The complexity of the candidate diagram of $n$ sites in the plane 
   is $O\pth{n^4}$.  The \storage complexity of
   the candidate diagram is $\Omega\pth{n^2}$ in the worst case and
   $O\pth{n^5}$ in all cases.%
\end{lemma}

\begin{proof}
    The lower bound is easy, and is left as an exercise to the reader.
    A \naive upper bound of $O\pth{n^5}$ on the \storage complexity,
    follows because:
    \begin{inparaenum}[(i)]
        \item all possible pairs of sites induce together
        $\binom{n}{2}$ bisectors,
        \item the complexity of the arrangement of the bisectors is
        $O\pth{n^4}$, and
        \item the candidate set of each face in this arrangement might
        have $n$ elements inside.
    \end{inparaenum}
    \DCGVer{{\qed}}
\end{proof}%

We leave the problem of closing the gap between the upper and lower
bounds of \lemref{complexity} as an open problem for further research.

\subsection{Sampling model}
\seclab{sample:model}%

Fortunately, the situation changes dramatically when randomization is
involved.  Let $\SiteSet$ be a set of $n$ sites in the plane.  For
each site $\site \in \SiteSet$, a parametric point
$\attribB = \pth{\attribB_1, \ldots, \attribB_d}$ is sampled
independently from $[0,1]^d$, with the following constraint: each
coordinate $\attribB_i$ is sampled from a (continuous) distribution
$\Distribution_i$, independently for each coordinate.  In particular,
the sorted order of the $n$ parametric points by a specific coordinate
yields a uniform random permutation (for the sake of simplicity of
exposition we assume that all the values sampled are distinct).
% \end{defn}

Our main result shows that, under the above assumptions, both the
complexity and the \storage complexity of the candidate diagram are
near linear in expectation --- see \thmref{candid:complexity} for the
exact statement.

%%%%%%%%%%%%%%%%%%%%%%%%%%%%%%%%%%%%%%%%%%%%%%%%%%%%%%%%%%%%%%% 
%%%%%%%%%%%%%%%%%%%%%%%%%%%%%%%%%%%%%%%%%%%%%%%%%%%%%%%%%%%%%%% 

\section{Backward analysis with high probability}
\seclab{minima:whp}%
\apndlab{minima:whp}%

%\subsection{A short detour into backward analysis}
%\seclab{backward:analysis:a}

\emph{Randomized incremental construction} is a powerful technique
used by geometric algorithms. Here, one is given a set of elements
$\ElemSet$ (e.g., segments in the plane), and one is interested in
computing some structure induced by these elements (e.g., the vertical
decomposition formed by the segments). To this end, one computes a
random permutation $\Pi = \permut{ \element_1, \ldots, \element_n}$ of
the elements of $\ElemSet$, and in the $i$\th iteration one computes
the structure $V_i$ induced by the $i$\th prefix
$\Pi_i = \permut{ \element_1, \ldots, \element_i}$ of $\Pi$ by
inserting the $i$\th element $\element_i$ into $V_{i-1}$ (e.g., split
all the vertical trapezoids of $V_{i-1}$ that intersect $\element_i$,
and merge together adjacent trapezoids with the same floor and
ceiling).

In \emphi{backward analysis} one is interested in computing the
probability that a specific object in $V_i$ was actually
created in the $i$\th iteration (e.g., a specific vertical trapezoid
in the vertical decomposition $V_i$). If the object of interest is
defined by at most $b$ elements of $\Pi_i$ for some constant $b$,
then the desired quantity is the probability that $\element_i$ is one
of these defining elements, which is at most $b/i$. In some cases, the
sum of these probabilities, over the $n$ iterations, counts the number
of times certain events happen during the incremental construction.
However, this yields only a bound in expectation.  For a 
high-probability bound, one can not apply this argument directly, as there
is a subtle dependency leakage between the corresponding indicator
variables involved between different iterations. (Without going into 
details, this is because the defining sets of the objects of
interest can have different sizes, and these sizes depend on which
elements were used in the permutation in earlier iterations.)

% \begin{defn}
%     \deflab{property}%
%     
Let $\PntSet$ be a set of $n$ elements.  A \emphi{property}
$\Property$ of $\PntSet$ is a function that maps any subset $\SetX$ of
$\PntSet$ to a subset $\Property(\SetX)$ of $\SetX$.
% \end{defn}
% 
Intuitively the elements in $\Property(\SetX)$ have some desired
property with respect to $\SetX$ (for example, let $\SetX$ be a set of
points in the plane, then $\Property(\SetX)$ may be those points in
$\SetX$ who lie on the convex hull of $\SetX$).
The following corollary (implied by \lemref{property:e} below)
provides a high-probability bound for backward analysis, and while the
proof is an easy application of the Chernoff inequality, it
nevertheless significantly simplifies some classical results on
randomized incremental construction algorithms.
%See \RefSecMinimaWHP for a more detailed discussion and a proof.

\begin{corollary}
    \corlab{property}%
    Let $\PntSet$ be a set of $n$ elements, let $c>1$ and $k \ge 1$ be
    prespecified numbers, and let $\Property(\SetX)$ be a property
    defined over any subset $\SetX \subseteq \PntSet$.  Now, consider
    a uniform random permutation $\permut{\pnt_1, \ldots, \pnt_n}$ of
    $\PntSet$, and let $\PntSet_i = \brc{\pnt_1, \ldots, \pnt_i}$.
    Furthermore, assume that we have $\cardin{\PSetX{\PntSet_i}} \leq k$
    simultaneously for all $i$ with probability at least $1-n^{-c}$.
    Let $X_i$ be the indicator variable of the event
    $\pnt_i \in \Property(\PntSet_i)$.  Then, for any constant
    $\gamma \geq 2e$, we have
    \begin{displaymath}
        \Prob{\Bigl. \sum_{i=1}^n X_i > \gamma \cdot (2 k \ln
           n)} \leq n^{-\gamma k} + n^{-c}.
    \end{displaymath}
    (If for all $\PntSetA \subseteq \PntSet$ we have that
    $\cardin{\PSetX{\PntSetA}} \leq k$, then the additional error term
    $n^{-c}$ is not necessary.)
\end{corollary}

In the remainder of this section we prove \lemref{property:e}, from which the 
above corollary is derived, and provide some examples of its applications.  
However, the above corollary statement is all that is required to prove our main results, 
and so if desired the reader can skip directly to \secref{proxy:set}.

%\corref{property} itself is sufficient for our purposes in the rest of the 
%paper.  So those who are eager for the main result can skip the rest
%of \secref{minima:whp} and proceed to \secref{proxy:set}.

\subsection{A short detour into backward analysis}
\seclab{backward:analysis:b}

% SARIEL NEW START

We need the following easy observation.

\begin{lemma}
    \lemlab{stupid}%
    Let $\Event_1, \ldots, \Event_t$ be disjoint events and let
    $\EventB$ be another event, such that
    $\beta = \Prob{\EventB \mid \Event_i}$ is the same for all $i$.
    Then
    $\Prob{\EventB \mid \Event_1 \cup \cdots \cup \Event_t} =
    \beta$.
\end{lemma}

\begin{proof}
    We have
    \begin{math}
        \Prob{\EventB \cap \pth{ \cup_i \Event_i}}%
        = %
        \sum_i \Prob{\EventB \cap \Event_i}%
        =%
        \sum_i \Prob{\EventB \mid \Event_i} \Prob{\Event_i}%
%        =%
%        \sum_i \beta \Prob{\Event_i}%
        =%
        \beta \sum_i \Prob{\Event_i}.
    \end{math}
    Hence
    \begin{math}
        \Prob{\EventB \mid \cup_i \Event_i}%
        =%
        {\Prob{\bigl. \EventB \cap (\cup_i \Event_i) }}
        \allowbreak /%
        {\Prob{\bigl. \cup_i \Event_i}}%
%        =%
%        \pth{\bigl.\beta \sum_i \Prob{ \Event_i}}/%
%        \pth{\bigl.\sum_i\Prob{ \Event_i}}%
        = \beta.
    \end{math}
    \DCGVer{\qed}%
\end{proof}

\begin{lemma}
    \lemlab{property:indep}%
    Let $\PntSet$ be a set of $n$ elements, and let $k_1, \ldots, k_n$
    be $n$ fixed non-negative integers depending only on $\PntSet$.
    Let $\Property$ be a property of $\PntSet$ satisfying the
    following condition: if $\cardin{\PntSetA} = i$ then
    $\cardin{\PSetX{\PntSetA}} = k_i$.  Now, consider a uniform random
    permutation $\permut{\pnt_1, \ldots, \pnt_n}$ of $\PntSet$. For
    any $i$, let $\PntSet_i = \brc{\pnt_1, \ldots, \pnt_i}$ and
    let $X_i$ be an indicator variable of the event that
    $\pnt_i \in \PSetX{\PntSet_i}$.  Then, the variables $X_i$ are
    mutually independent, for all $i$.
\end{lemma}

\begin{proof}
    Let $\Event_i$ denote the event that
    $\pnt_i \in \Property(\PntSet_i)$.  It suffices to show that the
    events $\Event_1, \ldots, \Event_n$ are mutually independent.  The
    insight is to think about the sampling process of creating the
    random permutation $\permut{\pnt_1, \ldots, \pnt_n}$ in a
    different way.  Imagine we randomly pick a permutation of elements
    in $\PntSet$, and set the last element to be $\pnt_n$.  Next, pick
    a random permutation of the remaining elements of
    $\PntSet \setminus \brc{\pnt_n}$ and set the last element to be
    $\pnt_{n-1}$.  Repeat this process until the whole permutation is
    generated.  Observe that $\Event_j$ is determined before
    $\Event_i$ for any $j > i$.

    Now, consider arbitrary indices
    $1 \leq i_1 < i_2 < \ldots < i_\psi \leq n$.  Observe that by our
    thought experiment, when determining the $i_1$\th value in the
    permutation, the suffix $\permut{\pnt_{i_1+1},\ldots,\pnt_n}$ is
    fixed.  Moreover, the property defined on the remaining set of
    elements marks $k_{i_1}$ elements, and these elements are randomly
    permuted before determining the $i_1$\th value.  Therefore, for
    any fixed sequence $\sfx = \permut{\pnt_{i_1+1},\ldots,\pnt_n}$,
    we have for a random permutation $\Permut$ of $\PntSet$, that
    \begin{math}
        \Prob{ \bigl.  \Event_{i_1} \vert{\, \sfx }}%
        =%
        \Prob{ \bigl.  \Permut \in \Event_{i_1} \vert{\,
              \SuffixY{\Permut}{i+1} = \sfx }}%
        =%
        k_{i_1}/i_1,
    \end{math}
    where
    $\SuffixY{\Permut}{i+1} = \permut{\PermutE_{i+1}, \ldots,
       \PermutE_{n}}$ is the suffix of the last $n-i_1$ elements of
    $\Permut$.  This also readily implies that
    \begin{math}
        \Prob{ \bigl.  \Event_{i_1} }%
        = %
        \sum_{\sfx} \Prob{ \bigl.  \Event_{i_1} \vert{\, \sfx }}
        \Prob{\sfx}%
        =%
        (k_1/i_1)\sum_{\sfx} \Prob{\sfx}%
        =%
        k_1/i_1.%
    \end{math}
    Thus, for any $\sfx$, we have
    $\Prob{ \bigl.  \Event_{i_1} \vert{\, \sfx }} = \Prob{ \bigl.
       \Event_{i_1} }$.

    \medskip%
    \noindent%
    \textbf{Informal argument.} %
    We are intuitively done --- knowing that
    $\Event_{i_2} \cap \ldots \cap \Event_{i_\psi} $ happens only
    gives us some partial information about what the suffix of the
    randomly picked permutation might be. However, the above states
    that even full knowledge of this suffix does not affect the
    probability of $\Event_{i_1}$ happening, thus implying that
    $\Event_{i_1}$ is independent of the other events.

    \medskip%
    \noindent%
    \textbf{Formal argument.} %
    Let $\PermutSet$ be the set of all permutations of $\PntSet$ such
    that $\Event_{i_2} \cap \ldots \cap \Event_{i_\psi} $ happens.
    Observe that whether or not a specific permutation $\Permut$
    belongs to $\PermutSet$ depends only on the value of its suffix
    $\SuffixY{\Permut}{i_1+1}$ --- indeed, once
    $\SuffixY{\Permut}{i_1+1}$ is known, one can determine whether the
    events $\Event_{i_2},\ldots, \Event_{i_\psi}$ happen for
    $\Permut$.

    For any index $i$, let $\SeqSetY{\PntSet}{n-i}$ be the set of
    sequences of distinct elements of $\PntSet$ of length $n-i$.  For a
    sequence $\sfx \in \SeqSetY{\PntSet}{n-i_1}$, let
    $\PermutSet[ \sfx]$ be the set of all permutations of $\PermutSet$
    with the suffix $\sfx$ (this set might be empty). For two
    different suffixes $\sfx, \sfx' \in \SeqSetY{\PntSet}{n-i_1}$, the
    corresponding sets of permutations $\PermutSet[ \sfx]$ and
    $\PermutSet[ \sfx']$ are disjoint. As such, the family
    \begin{math}
        \Set{\bigl. \PermutSet[ \sfx ] }{ \sfx \in \SeqSetY{\PntSet}{n-i_1}}
    \end{math}
    is a partition of $\PermutSet$ into disjoint sets.  

    By the above, for any $\sfx \in \SeqSetY{\PntSet}{n-i_1}$, such
    that $\PermutSet[\sfx]$ is not empty, we have that
    \begin{math}
        \Prob{\bigl. \Permut \in \Event_{i_1} \vert{\, \Permut \in
              \PermutSet[\sfx]}}%
        =%
        \Prob{\bigl. \Event_{i_1} \vert{\, \sfx }}%
        =%
        \Prob{\bigl. \Event_{i_1}}.
    \end{math}
    By \lemref{stupid}, for any arbitrary suffix $\sfx'$ such that
    $\PermutSet[\sfx']$ is not empty, we have
    \begin{align*}
      \Prob{ \bigl. \Event_{i_1} \vert{\, \Event_{i_2} \cap \ldots \cap
      \Event_{i_\psi} }}%
      =%
      \Prob{ \bigl. \Event_{i_1} \vert{\, \cup_{\sfx} \PermutSet[\sfx] }}%
      =%
      \Prob{ \bigl.  \Event_{i_1}\vert\, \sfx' }%
      =%
      \Prob{ \bigl. \Event_{i_1}}.
    \end{align*}

    \noindent%
    \textbf{Putting things together.} %
    By induction, we now have
    \begin{align*}
        \Prob{\bigl. \Event_{i_1} \cap \ldots \cap \Event_{i_\psi} }%
        &=%
        \Prob{\Event_{i_1} \mid \Event_{i_2} \cap \ldots \cap
              \Event_{i_\psi} }%
        \Prob{\bigl. \Event_{i_2} \cap \ldots \cap \Event_{i_\psi} } \\
        &=%
        \Prob{\Event_{i_1}\bigl.} \Prob{\bigl. \Event_{i_2} \cap
           \ldots \cap \Event_{i_{\psi}} }%
        =%
        \prod_{j=1}^\psi \Prob{\bigl. \Event_{i_j}},
        % = \prod_{j=1}^t \frac{k}{i_j},
    \end{align*}
    which implies that the events are mutually independent.
    \DCGVer{{\qed}}
\end{proof}
% SARIEL NEW END SARIEL NEW END

\begin{lemma}
    \lemlab{property:e}%
    % 
    % \RefProofInAppendix{p:e} %
    % 
    Let $\PntSet$ be a set of $n \ge e^2$ elements, and $k \ge 1$ be a fixed
    integer depending only on $\PntSet$.  Let $\Property$ be a
    property of $\PntSet$.  Now, consider a uniform random permutation
    $\permut{\pnt_1, \ldots, \pnt_n}$ of $\PntSet$. For each $i$,
    denote $\PntSet_i = \brc{\pnt_1, \ldots, \pnt_i}$ and let $X_i$ be
    an indicator variable of the event $\pnt_i \in \PSetX{\PntSet_i}$.
    Then we have:
    \InHsienRadical{\begin{compactenum}[(a)]}
	\NotInHsienRadical{\begin{compactenum}[(A)]}
        \item If $\cardin{\PSetX{\PntSetA}} = k$ whenever
        $\cardin{\PntSetA} \ge k$, and
        $\cardin{\PSetX{\PntSetA}} = \cardin{\PntSetA}$ whenever
        $\cardin{\PntSetA} < k$, then for any $\gamma \geq 2e$,
        \[
        \Bigl.\Prob{\Bigl. \sum\nolimits_{i} X_i > \gamma \cdot (2k \ln
           n)} \leq n^{-\gamma k}.
        \]
        
        \item The bound in \NotInHsienRadical{(A)}\InHsienRadical{(a)} holds under a weaker condition: For
        all\, $\PntSetA \subseteq \PntSet$ we have
        $\cardin{\PSetX{\PntSetA}} \leq k$.

        \item An even weaker condition suffices: For a random
        permutation
        $\langle\pnt_1, \ldots, \allowbreak \pnt_n \rangle$ of
        $\PntSet$, assume $\cardin{\PSetX{\PntSet_i}} \leq \newK$ for
        all $i$, with probability $1-n^{-c}$, where
        $\newK= c' \cdot k$, $c$ is an arbitrary constant, and $c'>1$
        is a constant that depends only on $c$.  Then for any
        $\gamma \geq 2e$,
        \[
        \Bigl.\Prob{\Bigl. \sum\nolimits_{i} X_i > \gamma \cdot (2c' k \ln
           n)} \leq n^{-\gamma k} + n^{-c}.
        \]
    \end{compactenum}
\end{lemma}

\begin{proof}
    \NotInHsienRadical{(A)}\InHsienRadical{(a)} Let $\Event_i$ be the event that
    $\pnt_i \in \Property(\PntSet_i)$.  By \lemref{property:indep} the
    events $\Event_1, \ldots, \Event_n$ are mutually independent, and
    % \begin{align*}
    \begin{math}
        \Prob{ \bigl. \Event_{i} }%
        =%
        \cardin{\Property(\PntSet_i)}/{i}%
        =%
        \min\pth{\bigl. {k}/{i},1}.%
    \end{math}
    % \end{align*}
    Thus, we have when $n \ge e^2$,
    \begin{align*}
        % \begin{math}
        \mu = \Ex{\sum\nolimits_i X_i}%
        \leq%
        k + \sum_{k < i\le n} \frac{k}{i}%
%        \leq%
%        k\pth{\Bigl. 1 + \ln n + 1 - \ln k}%
        \leq%
        k\pth{ 2+ \ln n}%
        \leq%
        2 k\ln n.%
        % \end{math}
    \end{align*}
    For any constant $\delta \geq 2e$, by Chernoff's inequality, we
    have $\Prob{\bigl. \sum_i X_i > \delta \mu } < 2^{-\delta \mu}$.
    Therefore by setting $\delta = {\gamma (2 k \ln n)}/{\mu}$ (which
    is at least $2e$ by the assumption that $\gamma \ge 2e$), we have
    \begin{align*}
        \Prob{ \sum\nolimits_i X_i > \gamma (2 k \ln n)}%
      =%
        < 2^{-\gamma (2 k \ln n)} < n^{-\gamma k}.
    \end{align*}

    \NotInHsienRadical{(B)}\InHsienRadical{(b)} In order to extend the result using the weaker condition, we
    augment the given property $\Property$ to a new property
    $\Property'$ that holds for \emph{exactly} $k$ elements.  So, fix
    an arbitrary ordering $\prec$ on the elements of $\PntSet$.  Now
    given any set $\SetX$ with $\cardin{\SetX} \geq k$, if
    $\cardin{\Property(\SetX)} = k$ then let
    $\Property'(\SetX) = \Property(\SetX)$. Otherwise, add the
    $k - \cardin{\Property(\SetX)}$ smallest elements in
    $\SetX \setminus \Property(\SetX)$ according to $\prec$ to
    $\Property(\SetX)$, and let $\Property'(\SetX)$ be the resulting
    subset of size $k$.  We also set $\Property'(\SetX) = \SetX$ for
    all $\SetX$ with $\cardin{\SetX} < k$.  The new property
    $\Property'$ complies with the original condition. For any
    $\SetX$, $\Property(\SetX) \subseteq \Property'(\SetX)$, which
    implies that an upper bound on the probability that the $i$\th
    element is in the property set $\Property'$ is an upper bound on
    the corresponding probability for $\Property$.

    \smallskip

    \NotInHsienRadical{(C)}\InHsienRadical{(c)} We truncate the given property $\Property$ if needed, so that
    it complies with \NotInHsienRadical{(B)}\InHsienRadical{(b)}.  Specifically, fix an arbitrary ordering
    $\prec$ on the elements of $\PntSet$.  Given any set $\SetX$, if
    $\cardin{\Property(\SetX)} \leq \newK$ then
    $\Property'(\SetX) = \Property(\SetX)$. Otherwise,
    $\cardin{\Property(\SetX)} > \newK$, and set $\Property'(\SetX)$
    to be the first $\newK$ of $\Property(\SetX)$ according to
    $\prec$.  Clearly, the new property $\Property'$ complies with the
    condition in \NotInHsienRadical{(B)}\InHsienRadical{(b)}.  Let $\EventA$ denote the event
    $\Property'(\PntSet_i) = \Property(\PntSet_i)$, for all $i$.  By
    assumption, we have
    \begin{math}
        \Prob{\EventA} \geq 1 - n^{-c}.
    \end{math}
    Similarly, let $\EventB$ be the event that
    $\sum_i X_i > \gamma (2 \newK \ln n)$. We now have that
    \begin{align*}
        \Prob{\EventB}%
        \leq%
        \Prob{ \EventB \sep{ \EventA }}%
        \Prob{ \EventA } %
        + \Prob{ \overline{\EventA}}%
        <%
        \pth{1 - \frac{1}{n^c}} n^{-\gamma \newK} + n^{-c}%
        \leq%
        n^{-\gamma k} + n^{-c}
    \end{align*}
    for any $\gamma \ge 2e$.
    \DCGVer{{\qed}}
\end{proof}

The result of \lemref{property:e} is known in the context of
randomized incremental construction algorithms (see
\cite[\S6.4]{bcko-cgaa-08}). However, the known proof is more
convoluted --- indeed, if the property $\Property(\SetX)$ has different
sizes for different sets $\SetX$, then it is no longer true that
variables $X_i$ in the proof of \lemref{property:e} are
independent. Thus the padding idea in part \NotInHsienRadical{(B)}\InHsienRadical{(b)} of the proof is crucial
in making the result more widely applicable.

\myparagraph{Example.} %
To see the power of \lemref{property:e} we provide two easy
applications --- both results are of course known, and are included
here to make it clearer in what settings \lemref{property:e} can be
applied. The impatient reader is encouraged to skip this example.
\InHsienRadical{\begin{enumerate}[(a)]}
\NotInHsienRadical{\begin{enumerate}[(A)]}
    \item \textbf{QuickSort}: We conceptually can think about
    QuickSort as being a randomized incremental algorithm, building
    up a list of numbers in the order they are used as pivots.
    Consider the execution of QuickSort when sorting a set $\PntSet$
    of $n$ numbers.  Let $\permut{\pnt_1, \ldots, \pnt_n}$ be the
    random permutation of the numbers picked in sequence by
    QuickSort.  Specifically, in the $i$\th iteration, it randomly
    picks a number $\pnt_i$ that was not handled yet, pivots based on
    this number, and then recursively handles the subproblems.
    At the $i$\th iteration, a set
    $\PntSet_i = \brc{\pnt_1, \ldots, \pnt_i}$ of pivots has already
    been chosen by the algorithm.  Consider a specific element
    $\query \in \PntSet$.  For any subset $\SetX \subseteq \PntSet$,
    let $\Property(\SetX)$ be the two numbers in $\SetX$ having
    $\query$ in between them in the original ordering of $\PntSet$ and
    are closest to each other.  In other words, $\Property(\SetX)$
    contains the (at most) two elements that are the endpoints of the
    interval of $\Re \setminus \SetX$ that contains $\query$.  Let
    $X_i$ be the indicator variable of the event
    $\pnt_i \in \Property(\PntSet_i)$ --- that is, $\query$ got
    compared to the $i$\th pivot when it was inserted.  Clearly, the
    total number of comparisons $\query$ participates in is
    $\sum_i X_i$, and by \lemref{property:e} the number of such
    comparisons is $O\pth{ \log n }$, with high probability, implying
    that QuickSort takes $O\pth{n \log n}$ time, with high
    probability.
    
    \item \textbf{Point-location queries in a history dag}: Consider a
    set of lines in the plane, and build their vertical decomposition
    using randomized incremental construction. Let
    $L_n = \permut{\ell_1, \ldots, \ell_n}$ be the permutation used by
    the randomized incremental construction.  Given a query point
    $\pnt$, the point-location time is the number of times the
    vertical trapezoid containing $\pnt$ changes in the vertical
    decomposition of $L_i= \permut{\ell_1, \ldots, \ell_i}$, as $i$
    increases. Thus, let $X_i$ the indicator variable of the event
    that $\ell_i$ is one of the (at most) four lines defining the
    vertical trapezoid containing $\pnt$ the vertical decomposition of
    $L_i$.  Again, \lemref{property:e} implies that the query time is
    $O\pth{ \log n}$, with high probability.  This result is well
    known, see \cite{cms-frric-93} and \cite[\S6.4]{bcko-cgaa-08}, 
    but our proof is arguably more direct and cleaner.
\end{enumerate}

%%%%%%%%%%%%%%%%%%%%%%%%%%%%%%%%%%%%%%%%%%%%%%%%%%%%%%%%%%%%%%% 
%%%%%%%%%%%%%%%%%%%%%%%%%%%%%%%%%%%%%%%%%%%%%%%%%%%%%%%%%%%%%%% 

\section{The proxy set}
\seclab{proxy:set}

Providing a reasonable bound on the complexity of the candidate
diagram directly seems challenging.  Therefore, we instead define for
each point $\query$ in the plane a slightly different set, called the
\emph{proxy set}.  First we prove that the proxy set for each point in
the plane has small size (see \lemref{proxy:set:size} below). Then we
prove that, with high probability, the proxy set of $\query$ contains
the candidate set of $\query$ for all points $\query$ in the plane
simultaneously (see \lemref{candidate:in:proxy} below).

\subsection{Definitions}
\seclab{proxy:set:def}

As before, the input is a set of sites $\SiteSet$.  For each site
$\site \in \SiteSet$, we randomly pick a parametric point
$\attribB \in [0,1]^d$ according to the sampling method described in
\secref{sample:model}.

\myparagraph{Volume ordering.} %
Given a point $\pnt = (\pnt_1, \ldots, \pnt_d)$ in $[0,1]^d$, the
\emphi{point volume} $\pv\pth{\pnt}$ of point $\pnt$ is defined to be
$\pnt_1 \pnt_2 \cdots \pnt_d$. That is, the volume of the
hyperrectangle with $\pnt$ and the origin as a pair of opposite
corners.
When $\pnt$ is specifically the associated parametric point of an
input site $\site$, we refer to the point volume of $\pnt$ as the
\emphi{parametric volume} of $\site$.
% \end{defn}
% 
Observe that if point $\pnt$ dominates another point $\pntA$ then
$\pnt$ must have smaller point volume (that is, $\pnt$ lies in the
hyperrectangle defined by $\pntA$).

% \begin{defn}
The \emphi{volume ordering} of sites in $\SiteSet$ is a permutation
$\permut{\site_1, \ldots, \site_n}$ ordered by increasing parametric
volume of the sites. That is,
$\pv(\attribB_1) \leq \pv(\attribB_2) \leq \ldots \leq \pv(\attribB_n)$,
where $\attribB_i$ is the parametric point of $\site_i$.
If $\attribB_i$ dominates $\attribB_j$ then $\site_i$ precedes
$\site_j$ in the volume ordering.  So if we add the sites in volume
ordering, then when we add the $i$\th site $\site_i$ we can ignore all
later sites when determining its region of influence --- that is, the
region of points whose candidate set $\site_i$ belongs to --- as no
later site can dominate $\site_i$.

\myparagraph{\InHsienRadical{\boldmath}$k$ Nearest neighbors.} %
For a set of sites $\SiteSet$ and a point $\query$ in the plane, let
{$\kdistSet{k}{\query}{\SiteSet}$} denote the \emphi{$k$\th nearest
   neighbor distance} to $\query$ in $\SiteSet$. That is, the $k$\th
smallest value in the multiset
$\Set{\bigl. \dist{\query}{\site}}{ \site \in \SiteSet}$.
The~\emphi{$k$ nearest neighbors} to $\query$ in $\SiteSet$ is the set
\begin{align*}
  \atmostK{\SiteSet}{k}{\query}%
  =%
  \Set{ \site \in \SiteSet }%
  { \dist{ \query}{ \site } \leq \kdistSet{k}{\query}{\SiteSet}
  \Big.}.
\end{align*}

\begin{defn}
    \deflab{proxy:set}%
    Let $\SiteSet$ be a set of sites in the plane, and let
    $\permut{\site_1, \ldots, \site_n}$ be the volume ordering of
    $\SiteSet$.  Let $\SiteSet_i$ denote the underlying set of the
    $i$\th prefix $\permut{\site_1, \ldots, \site_i}$ of
    $\permut{\site_1, \ldots, \site_n}$.  For a parameter $k$ and a
    point $\query$ in the plane, the \emphi{$k$\th proxy set} of
    $\query$ is the set of sites %
    \begin{align*}
      \proxySet{k}{\query}{\SiteSet}%
      =%
      \bigcup_{i=1}^n \atmostK{ \SiteSet_i}{k}{\query}.
    \end{align*}
    In words, site $\site$ is in $\proxySet{k}{\query}{\SiteSet}$ if
    it is one of the $k$ nearest neighbors to point $\query$ in some
    prefix of the volume ordering $\permut{\site_1, \ldots, \site_n}$.
\end{defn}

\subsection{Bounding the size of the proxy set}
\seclab{proxy:set:size}

The desired bound now follows by using backward analysis and
\corref{property}.

%\LemmaProofInFullVer{%
\begin{lemma}
   \lemlab{proxy:set:size}%
   Let $\SiteSet$ be a set of $n$ sites in the plane, and let
   $k \geq 1$ be a fixed parameter.  Then we have
   $\cardin{ \proxySet{k}{\query}{\SiteSet} } = \Owhp(k \log n)$
   simultaneously for all points $\query$ in the plane.
\end{lemma}
%}%
% 
\begin{proof}
    Fix a point $\query$ in the plane.  A site $\site$ gets added to
    the proxy set $\proxySet{k}{\query}{\SiteSet}$ if site $\site$ is
    one of the $k$ nearest neighbors of $\query$ among the underlying
    set $\SiteSet_i$ of some prefix of the volume ordering of
    $\SiteSet$.  Therefore a direct application of \corref{property}
    implies (by setting $\Property\pth{\SiteSet_i}$ to be
    $\atmostK{\SiteSet_i}{k}{\query}$), with high probability, that
    $\cardin{ \proxySet{k}{\query}{\SiteSet} \bigl. } = O\pth{ k \log
       n}$.
    
    Furthermore, this holds for all points in the plane
    simultaneously.  Indeed, consider the arrangement determined by
    the $\binom{n}{2}$ bisectors formed by all the pairs of sites in
    $\SiteSet$.  This arrangement is a simple planar map with
    $O\pth{n^4}$ vertices and $O\pth{n^4}$ faces.  Observe that within
    each face the proxy set cannot change since all points in this
    face have the same ordering of their distances to the sites in
    $\SiteSet$.  Therefore, picking a representative point from each
    of these $O\pth{n^4}$ faces, applying the high-probability bound
    to each of them, and then the union bound implies the
    claim.
    \DCGVer{{\qed}}
\end{proof}

\subsection{The proxy set contains the candidate set}
\seclab{proxy:contain:beer}

% The following corollary is implied by a careful (but straightforward)
% integration argument (see \apndref{h:p:minima}).
% % 
% \begin{corollary}
%     \corlab{integral:weak}%
%     % 
%     Let $\VolUArea{d}{\Delta}$ be the volume of the set of points
%     $\pnt$ in $[0,1]^d$ such that the point volume $\pv(\pnt)$ is at
%     most $\Delta$. That is,
%     \begin{align*}
%         \VolUArea{d}{\Delta} = %
%         \volX{ \Set{\pnt \in [0,1]^d}{\pv(\pnt) \le \Delta} }.
%     \end{align*}
%     Then for $\Delta \geq (\log n)/n$ we have
%     $\VolUArea{d}{\Delta} = \Theta(\Delta \log^{d-1} n)$.  In
%     particular,
%     $\VolUArea{d}{\log n/n} = \Theta( (\log^d n) / n )$.
% \end{corollary}

The following corollary is implied by a careful (but straightforward)
integration argument.
\begin{corollary}[Proof in \apndref{h:p:minima}]
    \corlab{integral:weak}%
    Let $\VolUArea{d}{\Delta}$ be the volume of the set of points
    $\pnt$ in $[0,1]^d$ such that the point volume $\pv(\pnt)$ is at
    most $\Delta$, where $\Delta \in (0,1)$. That is,
    \begin{displaymath}
        \Bigl.\VolUArea{d}{\Delta} = %
        \volX{ \Set{\pnt \in [0,1]^d}{\pv(\pnt) \le \Delta} }.
    \end{displaymath}
    Then, we have that
    \begin{math}
        \VolUArea{d}{\Delta} = \sum_{i=0}^{d-1} \frac{\Delta}{i!} {
           \ln^{i} \!\tfrac{1}{\Delta} }%
        =%
        O(\Delta \log^{d-1} n).
    \end{math}
\end{corollary}

%\LemmaProofInFullVer{%
\begin{lemma}
   \lemlab{candidate:in:proxy}%
   Let $\SiteSet$ be a set of $n$ sites in the plane, and let
   $k = \Theta(\log^d n)$ be a fixed parameter.  For all points
   $\query$ in the plane,  
   $\candid(\query) \subseteq \proxySet{k}{\query}{\SiteSet}$ 
   with high probability.%
\end{lemma}
%}

%\InNotProcVer{% 
\begin{proof}
    Fix a point $\query$ in the plane, and let $\site_i$ be any site
    \emph{not} in $\proxySet{k}{\query}{\SiteSet}$, and let
    $\attribB_i$ be the associated parametric point.  We claim that,
    with high probability, the site $\site_i$ is dominated by some
    other site which is closer to $\query$, and hence by the
    definition of dominating preference (\defref{dominating:pref}),
    $\site_i$ cannot be a site used by $\query$ (and thus
    $\site_i \notin \candid(\query)$).  Taking the union bound over
    all sites not in $\proxySet{k}{\query}{\SiteSet}$ then implies
    this claim.
    
    By \corref{integral:weak}, the total measure of the points in
    $[0,1]^d$ with point volume at most
    $\Delta = (\log n) / n$ is $O( (\log^d n)/n )$.  As such,
    by Chernoff's inequality, with high probability, there are
    $K = O(\log^d n)$ sites in $\SiteSet$ such that their parametric
    points have point volume smaller than $\Delta$.  In particular, by
    choosing $k$ to be larger than $K$, the underlying set
    $\SiteSet_k$ of the $k$\th prefix of the volume ordering of
    $\SiteSet$ will contain all these small point volume sites, and
    since $\SiteSet_k \subseteq \proxySet{k}{\query}{\SiteSet}$, so
    will $\proxySet{k}{\query}{\SiteSet}$.  Therefore, from this point
    on, we will assume that
    $\site_i \notin \proxySet{k}{\query}{\SiteSet}$ and
    $\Delta_i = \pv(\attribB_i) =\Omega(\log n / n)$.
    
    Now any site $\site$ with smaller parametric volume than $\site_i$
    is in the (unordered) prefix $\SiteSet_i$.  In particular,
    the $k$ nearest neighbors $\atmostK{\SiteSet_i}{k}{\query}$ of
    $\query$ in $\SiteSet_i$ all have smaller parametric volume than
    $\site_i$.  Hence $\proxySet{k}{\query}{\SiteSet}$ contains $k$
    points all of which have smaller parametric volume than $\site_i$,
    and which are closer to $\query$.  Therefore, the claim will be
    implied if one of these $k$ points dominates $\site_i$.

    \begin{figure}[t]
        \centerline{\includegraphics[width=0.3\linewidth]{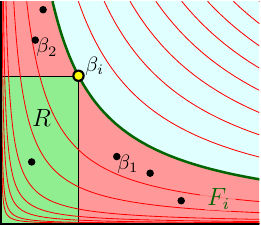}}%
        \captionof{figure}{The red shaded region $F_i$ consists of parametric points whose point volume is at most the point volume of $\attribB_i$.  The green shaded region $R$ consists of the parametric points that dominate $\attribB_i$.  The red curves are contours of the point volume function.}%
        \figlab{dominating}%
    \end{figure}

    The probability of a site $\site$ (that is closer to $\query$ than
    $\site_i$) with parametric point $\attribB$ to dominate $\site_i$
    is the probability that $\attribB \preceq \attribB_i$ given that
    $\attribB \in F_i$, where
    $F_i = \Set{ \attribB \in {[0,1]^d}}{ \pv(\attribB) \leq
       \Delta_i}$.  \corref{integral:weak} implies that
    $\volX{F_i} = \VolUArea{d}{\Delta_i} = O( \Delta_i
    \smash{\,\log^{d-1}} n )$.  The probability that a random
    parametric point in $[0,1]^d$ dominates $\attribB_i$ is exactly
    $\Delta_i$, and as such the desired probability
    $\Prob{ \attribB \preceq \attribB_i \sep{ \attribB \in F_i}}$ is
    equal to $\Delta_i / \VolUArea{d}{\Delta_i}$, which is
    $\Omega(1/\log^{d-1} n)$.  This is depicted in \figref{dominating}
    --- the probability of a random point picked uniformly from the
    region $F_i$ under the curve $y = \Delta_i/x$, induced by
    $\site_i$, to fall in the rectangle $R$.
    
    As the parametric point of each one of the $k$ points in
    $\atmostK{\SiteSet_i}{k}{\query}$ has equal probability to be
    anywhere in $F$, this implies the expected number of points in
    $\atmostK{\SiteSet_i}{k}{\query}$ which dominate $\site_i$ is
    $\Prob{ \attribB \preceq \attribB_i \sep{ \attribB \in F_i}} \cdot k =
    \Omega(\log n)$.
    Therefore by making $k$ sufficiently large, Chernoff's inequality
    implies the desired result.
    
    It follows that the statement holds, for all points in the plane
    simultaneously, by following the argument used in the proof of
    \lemref{proxy:set:size}.
    \DCGVer{{\qed}}
\end{proof}%
% }

%%%%%%%%%%%%%%%%%%%%%%%%%%%%%%%%%%%%%%%%%%%%%%%%%%%%%%%%%%%%%%% 
%%%%%%%%%%%%%%%%%%%%%%%%%%%%%%%%%%%%%%%%%%%%%%%%%%%%%%%%%%%%%%% 

\section{Bounding the complexity of the \InHsienRadical{\boldmath}$k$\th%
order proxy diagram}
\seclab{k:th:order}

The \emphi{$k$\th proxy cell} of $\query$ is the set of all
the points in the plane that have the same $k$\th proxy set associated
with them. Formally, this is the set
\[
  \Set{ \bigl. \pnt \in \Re^2}{ \bigl. \proxySet{k}{\pnt}{\SiteSet} =
  \proxySet{k}{\query}{\SiteSet} }\!.
\]
The decomposition of the plane into these faces is the 
\emphi{$k$\th order proxy diagram}.  
In this section, our goal is to prove that
the expected total diagram complexity of the $k$\th order proxy
diagram is $O\pth{ k^4 n\log n }$.  To this end, we relate this
complexity to the overlay of star-shaped polygons that rise out of the
$k$\th order Voronoi diagram.

\subsection{Preliminaries}
\seclab{k:th:order:prelim}

\subsubsection{The \InHsienRadical{\boldmath}$k$\th order Voronoi diagram}

Let $\SiteSet$ be a set of $n$ sites in the plane.  The
\emphi{$k$\th order Voronoi diagram} of $\SiteSet$ is a
partition of the plane into faces such that each cell is the locus of
points which have the same set of $k$ nearest sites in $\SiteSet$ (the
internal ordering of these $k$ sites, by distance to the query point,
may vary within the cell).  It is well known that the worst case
complexity of this diagram is $\Theta\pth{k(n-k)}$ (see \cite[\S6.5]{akl-vddt-13}).

\subsubsection{Arrangements of planes and lines}

One can interpret the $k$\th order Voronoi diagram in terms of an
arrangement of planes in $\Re^3$.  Specifically, ``lift'' each site to
the paraboloid $\pth{x,y, -(x^2+y^2)}$.  Consider the arrangement of
planes $\HPlanes$ tangent to the paraboloid at the lifted locations of
the sites.  A point on the union of these planes is of \emph{level
   $k$} if there are exactly $k$ planes strictly below it.  The
\emphi{$k$-level} is the closure of the set of points of
level $k$.\footnote{The lifting of the sites to the paraboloid
   $z = -(x^2+y^2)$ is done so that the definition of the $k$-level
   coincide with the standard definition. }
(For any set of $n$ hyperplanes in $\Re^d$, one can define $k$-levels
of arrangement of hyperplanes analogously.)
Consider a point $\query$ in the $xy$-plane.  The decreasing
$z$-ordering of the planes vertically below $\query$ is the same as
the ordering, by decreasing distance from $\query$, to the
corresponding sites.
Hence, let {$\planeLevel{k}{\HPlanes}$} denote the set of edges in the
arrangement $\HPlanes$ on the $k$-level, where an edge is a maximal
portion of the $k$-level that lies on the intersection of two planes
(induced by two sites).  Then the projection of the edges in
$\planeLevel{k-1}{\HPlanes}$ onto the $xy$-plane results in the edges
of the $k$\th order Voronoi diagram.
% The set of all the edges from the first $k$ levels is denoted by
% $\planeLevel{\leq k}{\SiteSet}$.
When there is no risk of confusion, we also use
$\planeLevel{k}{\SiteSet}$ to denote the set of edges in
$\planeLevel{k}{\HPlanes}$, where $\HPlanes$ is obtained by lifting
the sites in $\SiteSet$ to the paraboloid and taking the tangential
planes, as described above.

We also need the notion of $k$-levels of arrangement of lines.  For 
set of lines $L$ in the plane, let $\planeLevel{k}{L}$ denote the
set of edges in the arrangement of $L$ on the $k$-level. 
We need the following lemma.

\begin{lemma}
    \lemlab{on:L}%
    Let $L$ be a set of lines in general position in the plane, and let
    $\ell$ be any line in $L$.  Then at most $k+2$ edges from
    $\planeLevel{k}{L}$, the $k$-level of the arrangement of $L$, can
    lie on $\ell$.%
\end{lemma}
\begin{proof}
    This lemma is well known, and its proof is included here for the
    sake of completeness.
    
    \NotDCGVer{%
       \parpic[r]{ \includegraphics{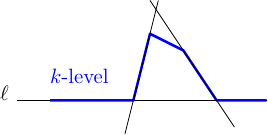} }%
    }
    \DCGVer{%
       \begin{figure}[t]
           \centerline{ \includegraphics{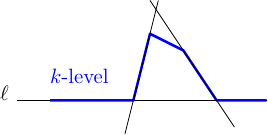} }%
           \captionof{figure}{}%
           \figlab{k:level}%
       \end{figure}%
    }

    \NotDCGVer{\InHsienRadical{\picskip{5}\linenumbers}}
    \DCGVer{\InHsienRadical{\picskip{6}\linenumbers}}
    \noindent%
    Perform a linear transformation such that $\ell$ is horizontal and
    the $k$-level is preserved.  As we go from left to right along the
    now horizontal line $\ell$ (starting from $-\infty$), we may leave
    and enter the $k$-level multiple times.  However, every time we
    leave and then return to the $k$-level we must intersect a
    negative slope line in between.  Specifically, both when we leave
    and return to the $k$-level, there must be an intersection with
    another line.  If the line intersecting the leaving point has a
    negative slope then we are done, so assume it has positive slope.
    In this case the level on $\ell$ decreases as we leave the
    $k$-level, and therefore when we return to the $k$-level, the
    point of return must be at an intersection with a negative slope
    line %
    \NotDCGVer{(see figure on the right).}%
    \DCGVer{(see \figref{k:level}).}%
    
    So after leaving and returning to the $k$-level $k+1$ times, there
    must be at least $k+1$ negative slope lines below, which implies
    that the remaining part of $\ell$ is on level strictly larger than
    $k$.
    \DCGVer{{\qed}}
\end{proof}

\begin{lemma}
    \lemlab{exact:k}%
    Let $L$ be a set of $n$ lines in general position in the
    plane. Fix any arbitrary insertion ordering of the lines in $L$,
    then the total number of distinct vertices on the $k$-level of the
    arrangement of $L$ seen over all iterations of this insertion
    process is bounded by $O(nk)$.%
\end{lemma}
\begin{proof}%
    Let $\ell_i$ be the $i$\th line inserted, and let $L_i$ be the set
    of the first $i$ lines.  Any new vertex on the $k$\th
    level created by the insertion must lie on $\ell_i$.  However, by
    \lemref{on:L} at most $k+2$ edges from $\planeLevel{k}{L_i}$ can
    lie on $\ell_i$.  As each such edge has at most two endpoints, the
    insertion of $\ell_i$ contributes $O\pth{ k }$ vertices to the
    $k$-level.  The bound now follows by summing over all $n$ lines.
    \DCGVer{{\qed}}
\end{proof}%

\subsection{Bounding the size of the below conflict-lists}
\seclab{below:conflict:size}

\subsubsection{The below conflict lists}

Let $\HPlanes$ be a set of $n$ planes in general position in $\Re^3$.
(For example, in the setting of the $k$\th order Voronoi diagram,
$\HPlanes$ is the set of planes that are tangent to the paraboloid at
the lifted locations of the sites.)  For any subset
$\RSample \subseteq\HPlanes$, let $\VLvl{k}{\RSample}$ denote the
vertices on the $k$-level of the arrangement of ${\RSample}$.
Similarly, let
$\VLvl{\leq k}{\RSample} = \bigcup_{i=0}^k \VLvl{k}{\RSample}$
be the set of vertices of level at most $k$ in the arrangement of
${\RSample}$, and let $\planeLevel{\leq k}{\RSample}$ be the set of
edges of level at most $k$ in the arrangement of ${\RSample}$.  For a
vertex $\vertex$ in the arrangement of ${\RSample}$, the \emphi{below
   conflict list} $\BclX{\vertex}$ of $\vertex$ is the set of planes
in $\HPlanes$ (not $\RSample$) that lie strictly below $\vertex$, and
let $\bclX{\vertex} = \cardin{\BclX{\vertex}}$.  For an edge $\edge$
in the arrangement of ${\RSample}$, the \emphi{below conflict list}
$\BclX{\edge}$ of $\edge$ is the set of planes in $\HPlanes$ (again,
not $\RSample$) which lie below $\edge$ (that is, there is at least one
point on $\edge$ that lies above such a plane), and let
$\bclX{\edge} = \cardin{\BclX{\edge}}$.  Our purpose here is to bound
the quantities
\begin{math}
    \smash{ \ExChar \big[ {\sum_{\vertex \in \VLvl{\leq k}{\RSample} }
          \bclX{\vertex}} \big] }
\end{math}
and
\begin{math}
    \smash{ \ExChar \big[ {\sum_{\edge \in \planeLevel{\leq
                k}{\RSample} } \bclX{\edge}} \big] }.
\end{math}

\subsubsection{The Clarkson-Shor technique}

In the following, we use the Clarkson-Shor technique
\cite{cs-arscg-89}, stated here without proof (see \cite{h-gaa-11} for
details).  Specifically, let $\SiteSet$ be a set of elements such that
any subset $\RSample \subseteq \SiteSet$ defines a corresponding set
of objects $\ObjX{\RSample}$ (e.g., $\SiteSet$ is a set of planes and
any subset $\RSample \subseteq \SiteSet$ induces a set of vertices in
the arrangement of planes $\RSample$).
Each potential object, $\object$, has a defining set and a stopping
set.  The \emphi{defining set}, $\DefSet{\object}$, is a subset of
$\SiteSet$ that must appear in $\RSample$ in order for the object to
be present in $\ObjX{\RSample}$.  We require that the defining set has
at most a constant size for every object.  The \emphi{stopping set},
$\KillSet{\object}$, is a subset of $\SiteSet$ such that if any of its
member appear in $\RSample$ then $\object$ is not present in
$\ObjX{\RSample}$.  We also naturally require that
$\KillSet{\object} \cap \DefSet{\object} = \varnothing$ for all object
$\object$.  Surprisingly, this already implies the following.

\begin{theorem}[Bounded Moments \cite{cs-arscg-89}]%
    \thmlab{moments}%
    Using the above notation, let $\SiteSet$ be a set of $n$ elements,
    and let $\RSample$ be a random sample of size $r$ from $\SiteSet$.
    Let $f(\cdot)$ be a monotonically increasing function bounded by a
    polynomial (that is, $f(n) = n^{O(1)}$).  We have %
    \begin{align*}
      \Ex{\sum\nolimits_{\object\in \ObjX{\RSample}} f\pth{
      \Bigl. \cardin{\KillSet{\object}} }}%
      =%
      O\pth{ \Ex{\Bigl. \cardin{\ObjX{\RSample}} }
      f\pth{\frac{n}{r}} },
    \end{align*}%
    where the expectation is taken over random sample $\RSample$.
\end{theorem}

\subsubsection{Bounding the below conflict-lists}

\myparagraph{The technical challenge.} %
The proof of the next lemma is technically interesting as it does not
follow in a straightforward fashion from the Clarkson-Shor
technique. Indeed, the below conflict list is \emph{not} the standard
conflict list. Specifically, the decision whether a vertex $\vertex$
in the arrangement of ${\RSample}$ is of level at most $k$ is a
``global'' decision of $\RSample$, and as such the defining set of
this vertex is neither of constant size, nor unique, as required to
use the Clarkson-Shor technique. If this was the only issue, the
extension by Agarwal \etal~\cite{ams-cmfal-98} could handle this
situation. However it is even worse: a plane
$\plane \in \HPlanes \setminus \RSample$ that is below a vertex
$\vertex \in \VLvl{\leq k}{\RSample}$ is not necessarily conflicting
with $\vertex$ (that is, in the stopping set of $\vertex$) --- as its
addition to $\RSample$ will not necessarily remove $\vertex$ from
$\VLvl{\leq k}{\RSample \cup \brc{\plane}}$.

\myparagraph{The solution.} %
Since the standard technique fails in this case, we need to perform
our argument somehow indirectly.  Specifically, we use a second random
sample and then deploy the Clarkson-Shor technique on this smaller
sample --- this is reminiscent of the proof bounding the size of
$\VLvl{\leq k}{\HPlanes}$ by Clarkson-Shor \cite{cs-arscg-89}, and the
proof of the exponential decay lemma of Chazelle and Friedman
\cite{cf-dvrsi-90}.

\begin{lemma}
   \lemlab{b:c:l}%
   Let $k$ be a fixed constant, and let $\RSample$ be a random sample
   (without replacement) of size $r$ from a set of $\HPlanes$ of $n$
   planes in $\Re^3$, we have
   \begin{align*}
     { \Ex{ \Bigl. \smash{\sum\nolimits_{\vertex \in \VLvl{\leq k}{\RSample} }}
     \bclX{\vertex}} } =O\pth{ nk^3 }.
   \end{align*}
\end{lemma}

% \InNotProcVer{%
\begin{proof}
    For the sake of simplicity of exposition, let us assume that the
    sampling here is done by picking every element into the random
    sample $\RSample$ with probability $r/n$. Doing the computations
    below using sampling without replacement (so we get the exact
    size) requires modifying the calculations so that the
    probabilities are stated using binomial coefficients --- this makes
    the calculation messier, but the results remain the same. See
    \cite{s-cstre-01} for further discussion of this minor issue.
    
    Fix a random sample $\RSample$. Now sample once again by picking
    each plane in $\RSample$, with probability $1/k$, into a subsample
    $\RSample'$.  Let us consider the probability that a vertex
    $\vertex \in \VLvl{\leq k}{\RSample}$ ends up on the lower
    envelope of $\RSample'$.  A lower bound can be achieved by the
    standard argument of Clarkson-Shor.  Specifically, if a vertex
    $\vertex$ is on the lower envelope then its three defining planes
    must be in $\RSample'$. Moreover, as
    $\vertex \in \VLvl{\leq k}{\RSample}$, by definition there are at
    most $k$ planes below $v$ that must not be in $\RSample'$.  So let
    $X_\vertex$ be the indicator variable of whether $\vertex$ appears
    on the lower envelope of $\RSample'$.  We then have
    \begin{align*}
      \displaystyle
      \ExOverCond{\RSample'}{\bigl. X_\vertex}{\RSample} \geq
      \frac{1}{k^3}(1-1/k)^k\geq \frac{1}{e^2k^3}.
    \end{align*}
    Observe that
    \begin{align}
      \DCGVer{\nonumber}%
      \ExOver{\RSample'} {\sum\nolimits_{\vertex \in
      \VLvl{0}{\RSample'}} \bclX{\vertex} }%
      &=%
        \ExOver{\RSample} {\biggl.\ExOverCond{\RSample'}%
        { \textstyle \sum_{\vertex \in \VLvl{0}{\RSample'}}
        \bclX{\vertex} } { \RSample } }%
        \DCGVer{%
      \\
      &%
        }%
        \geq%
        \ExOver{\RSample} {\biggl.\ExOverCond{\RSample'}%
        { \textstyle \sum_{\vertex \in \VLvl{\leq k}{\RSample}
        } X_\vertex \bclX{\vertex} } { \RSample } }.%
        \eqlab{amazing}
    \end{align}
    Fixing the value of $\RSample$, the lower bound above implies
    \begin{align*}
      \ExOverCond{\RSample'} {\sum\nolimits_{\vertex \in
      \VLvl{\leq k}{\RSample} } X_\vertex \bclX{\vertex}
      } {\RSample}%
      &=%
        \sum_{\vertex \in \VLvl{\leq k}{\RSample} }
        \ExOverCond{\RSample'}{\Bigl.  X_\vertex \bclX{\vertex} }
        {\RSample}%
        \DCGVer{
      \\&}%
          =%
          \sum_{\vertex \in \VLvl{\leq k}{\RSample} }
          \bclX{\vertex} \ExOverCond{\RSample'}{\Bigl.  X_\vertex
          }{\RSample}%
          \geq%
          \sum_{\vertex \in \VLvl{\leq k}{\RSample} }
          \frac{\bclX{\vertex}}{e^2k^3},
    \end{align*}
    by linearity of expectations and as $\bclX{\vertex}$ is a constant
    for $\vertex$.  Plugging this into \Eqref{amazing}, we have
    \begin{align}
        % \mu =
      \ExOver{\RSample'}{%
      \sum\nolimits_{\vertex \in \VLvl{0}{\RSample'}}\,
      \bclX{\vertex} }%
      \geq%
      \ExOver{\RSample} {%
      \sum\nolimits_{\vertex \in \VLvl{\leq k}{\RSample} }
      \frac{\bclX{\vertex}}{e^2k^3}%
      }%
      =%
      {\displaystyle \frac{1}{e^2k^3}}
      \ExOver{\RSample}{%
      \sum\nolimits_{\vertex \in \VLvl{\leq k}{\RSample}} \bclX{\vertex} }.%
      \eqlab{stupid:2}%
    \end{align}
    
    Observe that $\RSample'$ is a random sample of $\RSample$ which by
    itself is a random sample of $\HPlanes$. As such, one can
    interpret $\RSample'$ as a direct random sample of $\HPlanes$.
    The lower envelope of a set of planes has linear complexity, and
    for a vertex $\vertex$ on the lower envelope of $\RSample'$ the
    set $\BclX{\vertex}$ is the standard conflict list of
    $\vertex$. As such, \thmref{moments} implies
    \begin{align*}
        % \mu =
      \ExOver{\RSample'}{\sum\nolimits_{\vertex \in
      \VLvl{0}{\RSample'}} \bclX{ v }}%
      =%
      O\pth{\cardin{\RSample'} \cdot
      \displaystyle\frac{n}{\cardin{\RSample'}} }%
      =%
      O\pth{n }.
    \end{align*}
    Plugging this into \Eqref{stupid:2} implies the claim.
    \DCGVer{{\qed}}
\end{proof}%
% }

\begin{corollary}
    \corlab{co:b:c:l}%
    Let $\RSample$ be a random sample (without replacement) of size
    $r$ from a set $\HPlanes$ of $n$ planes in $\Re^3$. We have that
    \begin{math}
        \smash{ \ExOver{\RSample}{\sum_{\edge \in \planeLevel{\leq
                    k}{\RSample}} \bclX{\edge}} } = O\pth{ nk^3 }.
    \end{math}
\end{corollary}
\begin{proof}
    Under general position assumption every vertex in the arrangement
    of ${\HPlanes}$ is adjacent to $8$ edges. For an edge
    $\edge = \vertexA \vertex$, it is easy to verify that
    $\BclX{\edge} \subseteq \BclX{\vertexA} \cup \BclX{\vertex}$, and
    as such we charge the conflict list of $\edge$ to its two
    endpoints $u$ and $\vertex$, and every vertex get charged
    $O\pth{1}$ times.  Now, the claim follows by \lemref{b:c:l}.
    
    This argument fails to capture edges that are rays in the
    arrangement, but this is easy to overcome by clipping the
    arrangement to a bounding box that contains all the vertices of
    the arrangement. We omit the easy but tedious details.
    \DCGVer{{\qed}}
\end{proof}%

\subsection{Environments and overlays}
\seclab{enviroments:overlays}
% \myparagraph{Environments and overlays.} %
% 
For a site $\site$ in $\SiteSet$ and a constant $k$, the \emphi{$k$\!
   environment} of $\site$, denoted by $\envSet{k}{\site}{\SiteSet}$,
is the set of all the points in the plane such that $\site$ is one of
their $k$ nearest neighbors in $\SiteSet$: %
\begin{align*}
  \envSet{k}{\site}{\SiteSet} = \brc{ \query \in \Re^2%
  \sep{ \site \in \atmostK{\SiteSet}{k}{\query} }}\!.
\end{align*}

\parpic[r]{%
   % \begin{minipage}{3cm}
   \includegraphics[page=4,width=0.2\linewidth]{figs/k_order}%
   % \captionof{figure}{}
   % \figlab{k:order}
   % \end{minipage}
}

\NotDCGVer{\InHsienRadical{\picskip{6}\linenumbers}}
\DCGVer{\InHsienRadical{\picskip{6}\linenumbers}}
See the figure on the right for an example what this environment looks
like for different values of $k$.  One can view the $k$ environment of
$\site$ as the union of the $k$\th order Voronoi cells which have
$\site$ as one of the $k$ nearest sites.  Observe that the overlay of
the polygons
$\envSet{k}{\site_1}{\SiteSet}, \ldots, \envSet{k}{\site_n}{\SiteSet}$
produces the $k$\th order Voronoi diagram of $\SiteSet$.  It is also
known that each $k$ environment of a site is a star-shaped polygon
(see Aurenhammer and Schwarzkopf \cite{as-soria-92}).

\begin{lemma}
    \lemlab{star:shaped}%
    The set $\envSet{k}{\site}{\SiteSet}$ is a star-shaped polygon
    with respect to the point $\site$.
\end{lemma}
\begin{proof}
    Consider the set of all $n-1$ bisectors determined by $\site$ and
    any other site in $\SiteSet$.  For any point $\query$ in the
    plane, $\pnt\in \envSet{k}{\site}{\SiteSet}$ holds if the segment
    from $\site$ to $\pnt$ crosses at most $k-1$ of these bisectors.
    The star-shaped property follows as when walking along any ray
    emanating from $\site$, the number of bisectors crossed is a
    monotonically increasing function of distance from $\site$.
    Moreover, $\envSet{k}{\site}{\SiteSet}$ is a polygon as its
    boundary is composed of subsets of straight line bisectors.
    \DCGVer{{\qed}}
\end{proof}%

Going back to our original problem.  Let $k$ be a fixed constant, and
let $\permut{\site_1, \ldots, \site_n}$ be the
volume ordering of $\SiteSet$.  As usual, we use $\SiteSet_i$ to
denote the unordered $i$\th prefix of $\permut{\site_1, \ldots, \site_n}$.  Let
$\env_i = \envSet{k}{\site_i}{\SiteSet_i}$, the union of all
cells in the $k$\th order Voronoi diagram of $\SiteSet_i$ where
$\site_i$ is one of the $k$ nearest neighbors.
% participates in.

\begin{observation}
    % \obslab{env:overlay:proxy}%
    %
    The arrangement determined by the overlay of the polygons
    $\env_1, \ldots, \env_n$ is the $k$\th order proxy diagram of
    $\SiteSet$.
\end{observation}

\subsection{Putting it all together}
\seclab{all:together}

The proof of the following lemma is similar in spirit to the argument
of Har-Peled and Raichel \cite{hr-ecrwv-14}.

\begin{lemma}
    \lemlab{proxy:complexity}%
    Let $\SiteSet$ be a set of $n$ sites in the plane,
    let $\permut{\site_1, \ldots, \site_n}$ be the volume ordering of
    $\SiteSet$, and let $k$ be a fixed number.  The expected
    complexity of the arrangement determined by the overlay of the
    polygons $\,\env_1, \ldots, \env_n$ (and therefore, the expected
    complexity of the $k$\th order proxy diagram) is
    $ O\pth{ k^4 n \log n }$, where
    $\env_i = \envSet{k}{\site_i}{\SiteSet_i}$ and
    $\SiteSet_i = \brc{\site_1, \ldots, \site_i}$ is the underlying
    set of the $i$\th prefix of $\permut{\site_1, \ldots, \site_n}$,
    for each $i$.
\end{lemma}

\begin{proof}
    As the arrangement of the overlay of the polygons
    $\env_1, \ldots, \env_n$ is a planar map, it suffices to bound the
    number of edges in the arrangement.  Fix an iteration $i$, and
    observe that $\SiteSet_i$ is fixed once $i$ is fixed.
    For an edge $e \in \planeLevel{\leq k}{\SiteSet_i}$, let $X_\edge$
    be the indicator variable of the event that $\edge$ was created in
    the $i$\th iteration, and furthermore, lies on the boundary of
    $\env_i$.  Observe that $\Ex{X_\edge \mid \SiteSet_i} \leq 4/i$,
    as an edge appears for the first time in round $i$ only if one of
    its (at most) four defining sites was the $i$\th site inserted.

    For each $i$, let $\Edges(\env_i)$ be the edges in
    $\planeLevel{\leq k}{\SiteSet_i}$ that appear on the boundary of
    $\env_i$ (for simplicity we do not distinguish between edges in
    $\planeLevel{\leq k}{\SiteSet_i}$ in $\Re^3$ and their projection
    in the plane).  
    Created in the $i$\th iteration, an edge $\edge$
    in $\Edges(\env_i)$ is going to be broken into several pieces in
    the final arrangement of the overlay.  Let $n_\edge$ be the number
    of such pieces that arise from $\edge$. 

    Here we claim that $n_\edge \leq c\cdot k\bclX{\edge}$
    for some constant $c$.  Indeed, $n_\edge$ counts the number of future
    intersections of $\edge$ with the edges of $\Edges(\env_j)$, for
    any $j>i$.  
    As the edge $\edge$ is on the $k$-level at the time of
    creation, and the edges in $\Edges(\env_j)$ are on the $k$-level
    when they are being created (in the future), these edges must lie
    below $\edge$.  Namely, any future intersect on $\edge$ are caused
    by intersections of (pairs of) planes in $\BclX{\edge}$.  So
    consider the intersection of all planes in $\BclX{\edge}$ on the
    vertical plane containing $\edge$.  
    (Since $\SiteSet_i$ is fixed, $\BclX{\edge}$ is
    also fixed for all $\edge \in \planeLevel{\leq k}{\SiteSet_i}$.)
    On this vertical plane,
    $\BclX{\edge}$ is a set of $\bclX{\edge}$ lines, whose insertion
    ordering is defined by the suffix of the permutation
    $\permut{\site_{i+1}, \ldots, \site_{n}}$.  Now any edge of
    $\Edges(\env_j)$, for some $j>i$, that intersects $\edge$ must
    appear as a vertex on the $k$-level at some point during the
    insertion of these lines.  However, by \lemref{exact:k}, applied
    to the lines of $\BclX{\edge}$ on the vertical plane of $\edge$,
    under any insertion ordering there are at most
    $O\pth{ k \bclX{\edge} }$ vertices that ever appear on the
    $k$-level.

    Let
    \begin{math}
        Y_i = \sum_{\edge \in \Edges(\env_i)} n_\edge = \sum_{\edge
           \in \Edges_{\leq k}(\SiteSet_i) } n_\edge X_\edge
    \end{math}
    be the total (forward) complexity contribution to the final
    arrangement of edges added in round $i$.  We thus have
    \begin{align*}
        \ExCond{\Bigl. Y_i}{\SiteSet_i}%
        &=%
        \Ex{%
              \sum\nolimits_{\edge \in \planeLevel{\leq k} {\SiteSet_i}}%
           n_\edge X_\edge \sep{\biggl.\SiteSet_i}}%
        \leq%
        \ExCond{ \sum\nolimits_{\edge \in \planeLevel{\leq
                    k}{\SiteSet_i}} {ck \bclX{\edge} X_\edge}
        }{\SiteSet_i}%
        \\&%
        =%
        \sum\nolimits_{\edge \in \planeLevel{\leq k}{\SiteSet_i}} {ck
           \bclX{\edge} } \ExCond{\Bigl.X_\edge}{\SiteSet_i}%
        \leq%
        {\frac{4ck}{i} \cdot \sum\nolimits_{\edge \in \planeLevel{\leq
                 k}{\SiteSet_i}} \bclX{\edge} }.
    \end{align*}

    The total complexity of the overlay arrangement of the polygons
    $\env_1, \ldots, \allowbreak \env_n$ is asymptotically bounded by
    $\sum_i Y_i$, and so by \corref{co:b:c:l} we have
    \begin{align*}
      \Ex{\Bigl. \textstyle \sum_i Y_i}%
      &=%
        \sum_i \Ex{ \biggl. \ExCond{\Bigl.Y_i}{\SiteSet_i} }%
        \leq%
        \sum_i \Ex{ {\frac{4ck}{i}\cdot \sum\nolimits_{\edge \in
        \planeLevel{\leq k}{\SiteSet_i}} \bclX{\edge} }
        }%
        \DCGVer{
      \\&}%
          =%
          O\pth{\sum\nolimits_i \frac{nk^4}{i} }%
          =%
          O\pth{\Bigl. k^4 n \log n}. 
          \DCGVer{{\tag*\qed}}
    \end{align*}
\end{proof}

%%%%%%%%%%%%%%%%%%%%%%%%%%%%%%%%%%%%%%%%%%%%%%%%%%%%%%%%%%%%%%%%%%% 
%%%%%%%%%%%%%%%%%%%%%%%%%%%%%%%%%%%%%%%%%%%%%%%%%%%%%%%%%%%%%%%%%%% 

\section{On the expected size of the staircase}
\seclab{size:candidate:set}

\subsection{Number of staircase points}
\seclab{stair}

\subsubsection{The two dimensional case}
\seclab{stair:2d}

\begin{corollary}%
    \corlab{2:d}%
    % 
    % \RefProofInAppendix{2:d} %
    %
    Let $\PntSet$ be a set of $n$ points sampled uniformly at random
    from the unit square $[0,1]^2$.  Then the number of staircase
    points $\stair(\PntSet)$ in $\PntSet$ is $\Owhp(\log n)$.
\end{corollary}

\begin{proof}%:in:appendix:e}{\corref{2:d}}{2:d}
    If we order the points in $\PntSet$ by increasing $x$-coordinate,
    then the staircase points are exactly the points which have the
    smallest $y$-values out of all points in their prefix in this
    ordering.  As the $x$-coordinates are sampled uniformly at random,
    this ordering is a random permutation $\permut{y_1, \ldots, y_n}$
    of the $y$-values $\valsY$.  Let $X_i$ be the indicator variable
    of the event that $y_i$ is the smallest number in
    $\valsY_i = \brc{y_1, \ldots, y_i}$ for each $i$.  By setting
    property $\Property(\valsY_i)$ to be the smallest number in the
    prefix $\valsY_i$, we have $\sum_{i=1}^n X_i = O(\log n)$ with
    high probability by \corref{property}.
    \DCGVer{{\qed}}
\end{proof}%:in:appendix:e}

\subsubsection{Higher dimensions}
\seclab{stair:hd}

\begin{lemma}
    \lemlab{staircase:union}%
    % 
    % \RefProofInAppendix{staircase:union} %
    %
    Fix a dimension $d \ge 2$.  Let $m$ and $n$ be parameters, such
    that $m \leq n$.  Let
    $\OPntSet = \permut{\pntA_1, \ldots, \pntA_m}$ be an ordered set
    of $m$ points picked randomly from $[0,1]^d$ as described in
    \secref{sample:model}.  Assume that we have
    $\cardin{\stairX{\OPntSet_i}} = O(c_d \log^{d-1} n)$, with high
    probability with respect to $m$ for all $i$ simultaneously, where
    $\OPntSet_i = \brc{\pntA_1, \ldots, \pntA_i}$ is the underlying
    set of the $i$\th prefix of $\OPntSet$.  Then, the set
    $\stair = \bigcup_{i=1}^m \stairX{\OPntSet_i}$ has size
    $O(c_d \log^d n)$, with high probability with respect to $m$.
\end{lemma}
\begin{proof}
    Let $\cardin{\stairX{\OPntSet_i}} \leq c' \cdot c_d \ln^{d-1} n$
    with probability $1-m^{-c}$ for large enough constant $c$ and some
    constant $c'$ depending on $c$.  By setting
    $\Property(\OPntSet_i) = \stairX{\OPntSet_i}$, we have that
    $\Prob{\bigl.\cardin{\stair} > \gamma (2k \ln m)} \leq m^{-\gamma
       k} + m^{-c}$ for $\smash{k = c' \cdot c_d \ln^{d-1} n}$ and any
    $\gamma \geq 2e$, by \corref{property}.  Setting
    $\gamma = 2e\ln n / \ln m$
    % , we have
    % \[
    %     \Prob{\bigl. \cardin{\stair} > 4e c' \cdot c_d \ln^{d} n}
    %     \leq
    %     e^{-2ec' \cdot c_d \ln^{d} n} + m^{-c}
    %     \leq %e^{-2ec' \cdot c_d \ln^{d} m} + m^{-c},
    % \]
    % which
    implies the claim.
    \DCGVer{{\qed}}
\end{proof}

\begin{lemma}
    \lemlab{staircase}
    % 
    % \RefProofInAppendix{staircase} %
    %
    Fix a dimension $d \ge 2$.  Let $m, n$ be parameters, such that
    $m \leq n$.  Let $\PntSet$ be a set of $m$ points picked randomly
    from $[0,1]^d$ as described in \secref{sample:model}.  Then,
    $\cardin{\stairX{\PntSet}} = O( c_d \log^{d-1} n )$ holds, with
    high probability with respect to $m$, for some constant $c_d$ that
    depends only on $d$.
\end{lemma}
\begin{proof}%{\lemrefpage{staircase}}{staircase}
    The argument follows by induction on dimension.  The
    two-dimensional case follows from \corref{2:d}.  Assume we have
    proven the claim for all dimension smaller than $d$.
    
    Now, sort $\PntSet$ by increasing value of the $d$\th coordinate,
    and let $\pnt_i = \pth{\pntA_i, \ell_i}$ be the $i$\th point in
    $\PntSet$ in this order for each $i$, where $\pntA_i$ is a
    $(d-1)$-dimensional vector and $\ell_i$ is the value of the $d$\th
    coordinate of $\pnt_i$.  Observe that the points
    $\pntA_1, \ldots, \pntA_m$ are randomly, uniformly, and
    independently picked from the hypercube $[0,1]^{d-1}$.  Now, if
    $\pnt_i$ is a minima point of $\PntSet$, then it is a minima point
    of $\brc{\pnt_1,\ldots, \pnt_i}$.  But this implies that $\pntA_i$
    is a minima point of
    $\OPntSet_i = \brc{ \pntA_1, \ldots, \pntA_i}$ as well. Namely,
    $\pntA_i \in \stair = \bigcup_{i=1}^m \stairX{\OPntSet_i}$.  This
    implies that $\cardin{\stairX{\PntSet}} \leq \cardin{\stair}$.
    Now, applying induction hypothesis on each $\OPntSet_i$ in
    dimension $d-1$ we have
    $\cardin{\stairX{\OPntSet_i}} = O( c_{d-1} \log^{d-2} n )$ holds
    for all $i$, with high probability with respect to $m$.  Plugging
    it into \lemref{staircase:union} we have
    $\cardin{\stairX{\PntSet}} \leq \cardin{\stair} = O( c_{d-1}
    \log^{d-1} n )$, with high probability with respect to $m$.
    Choosing a proper constant $c_d$ now implies the claim.
    \DCGVer{{\qed}}
\end{proof}

\begin{lemma}
    \lemlab{staircase:w:h:p}%
    Fix a dimension $d \ge 2$.  Let
    $\OPntSet = \permut{\pntA_1, \ldots, \pntA_n}$ be an ordered set
    of $n$ points picked randomly from $[0,1]^d$ (as described in
    \secref{sample:model}), and
    $\OPntSet_i = \brc{\pntA_1, \ldots, \pntA_i}$ is the $i$\th
    (unordered) prefix of $\OPntSet$.  Then, the set
    $\bigcup_{i=1}^n \stairX{\OPntSet_i}$ is of size
    $\Owhp( c_d \log^{d} n)$, and the staircase $\stairX{\PntSet}$ is
    of size $\Owhp( c_d \log^{d-1} n)$.
\end{lemma}

\begin{proof}%:in:appendix:e}{\lemrefpage{staircase:w:h:p}}{s:w:h:p}
    By \lemref{staircase:union}, the set
    $\bigcup_{i=1}^n \stairX{\OPntSet_i}$ is of size
    $O( c_d \log^d n )$, with high probability.  By
    \lemref{staircase}, the set $\stairX{\PntSet}$ is of size
    $O( c_d \log^{d-1} n )$, with high probability.
    \DCGVer{{\qed}}
\end{proof}%:in:appendix:e}

\begin{remark}
    In the proof of \lemref{staircase} whether a point is on the
    staircase (or not) only depends on the coordinate orderings of the
    points and not their actual values.
    
    The basic recursive argument used in \lemref{staircase} was used
    by Clarkson \cite{c-enkci-04} to bound the expected number of
    $k$-sets for a random point set.  Here, using \corref{property}
    enables us to get a high-probability bound.
    
    Note that the definition of the staircase can be made with respect
    to any corner of the hypercube (that is, this corner would replace
    the origin in the definition dominance, point volume, the
    exponential grid, etc).  Taking the union over all $2^d$ such
    staircases gives us the subset of $\PntSet$ on the orthogonal
    convex hull of $\PntSet$.  Therefore \lemref{staircase:w:h:p} also
    bounds the number of input points on the orthogonal convex hull.
    As the vertices on the convex hull of $\PntSet$ are a subset of
    the points in $\PntSet$ on the orthogonal convex hull, the above
    also implies the same bound on the number of vertices on the
    convex hull.
\end{remark}

\subsection{Bounding the size of the candidate set}
\seclab{size:beer:set}

We can now readily bound the size of the candidate set for any point
in the plane.%

\begin{lemma}
    \lemlab{candid}%
    % 
    % \RefProofInAppendix{candid} %
    % 
    Let $\SiteSet$ be a set of $n$ sites in the plane, where for each
    site $\site$ in $\SiteSet$, a parametric point from a distribution
    over $[0,1]^d$ is sampled (as described in \secref{sample:model}).
    Then, the candidate set has size $\Owhp(\log^{d} n)$
    simultaneously for all points in the plane.
\end{lemma}

\begin{proof}
    Consider the arrangement of bisectors of all pairs of points of
    $\SiteSet$.  This arrangement has complexity $O\pth{ n^4 }$, and
    inside each cell the candidate set is the same.  Now for any point
    in a cell of the arrangement, \lemref{staircase:w:h:p} immediately
    gives us the stated bound, with high probability.  Therefore
    picking a representative point from each cell in this arrangement
    and applying the union bound imply the claim.  \DCGVer{{\qed}}
\end{proof}

%%%%%%%%%%%%%%%%%%%%%%%%%%%%%%%%%%%%%%%%%%%%%%%%%%%%%%%%%%%%%%% 
%%%%%%%%%%%%%%%%%%%%%%%%%%%%%%%%%%%%%%%%%%%%%%%%%%%%%%%%%%%%%%% 

\section{The main result}
\seclab{pf:candid:complexity}

We now use the bound on the complexity of the proxy diagram, as well
as our knowledge of the relationship between the candidate set and the
proxy set to bound the complexity, as well as the \storage complexity, 
of the candidate diagram.

Recall that the \emph{complexity} of a candidate diagram, treated as a
planar arrangement, is the total number of edges, faces, and vertices 
in the diagram. The \emph{\storage complexity} of the candidate 
diagram is 
%the total amount of memory needed to store the diagram explicitly, 
%and is bounded by the complexity of the candidate diagram together with 
the sum of the sizes of candidate sets over all the faces in the 
arrangement of the diagram.

\begin{theorem}
    \thmlab{candid:complexity}%
    % 
    % \RefProofInAppendix{c:c} %
    % 
    Let $\SiteSet$ be a set of $n$ sites in the plane, where for each
    site in $\SiteSet$ we sample an associated parametric point in
    $[0,1]^d$, as described in \secref{sample:model}.  Then, the
    expected complexity of the candidate diagram is
    $O\bigl(n\log^{8d+5} n \bigr)$. The expected \storage complexity
    of this candidate diagram is $O\bigl( n\log^{9d+5} n \bigr)$.
\end{theorem}

\begin{proof}
    Fix $k$ to be sufficiently large such that $k=\Theta(\log^d n)$.
    By \lemref{proxy:complexity} the expected complexity of the proxy
    diagram is $O(k^4 n \log n)$.  Triangulating each polygonal cell
    in the diagram does not increase its asymptotic complexity.
    \lemref{proxy:set:size} implies that, %with high probability, 
    the proxy set has size $\Owhp\pth{ k \log n }$ simultaneously for all the
    points in the plane.  Now, \lemref{candidate:in:proxy} implies
    that, with high probability, the proxy set contains the candidate
    set for any point in the plane.
    
    The resulting triangulation has $O(k^4 n \log n)$ faces, and
    inside each face all the sites that might appear in the candidate
    set are all present in the proxy set of this face.
    By \lemref{complexity}, the complexity of an $m$-site candidate
    diagram is $O(m^4)$.  Therefore the complexity of the candidate
    diagram per face is $\Owhp\pth{ (k \log n)^4 }$ %, with high probability
    (clipping the candidate diagram of these sites to the containing
    triangle does not increase the asymptotic complexity).
    Multiplying the number of faces, $O(k^4 n \log n)$, by the
    complexity of the arrangement within each face,
    $O\pth{ (k \log n)^4}$, yields the desired result.
    
    The bound on the \storage complexity follows readily from the
    bound on the size of the candidate set from \lemref{candid}.
    \DCGVer{{\qed}}
\end{proof}

%%%%%%%%%%%%%%%%%%%%%%%%%%%%%%%%%%%%%%%%%%%%%%%%%%%%%%%%%%%%%%% 
%%%%%%%%%%%%%%%%%%%%%%%%%%%%%%%%%%%%%%%%%%%%%%%%%%%%%%%%%%%%%%% 

\section*{Acknowledgments}

The authors would like to thank Pankaj Agarwal, Ken Clarkson, Nirman
Kumar, and Raimund Seidel for useful discussions related to this work.
We are also grateful to the anonymous SoCG reviewers for their helpful
comments.  \DCGVer{%
   Work on this paper was partially supported by NSF AF award
   CCF-1421231, and % Started June 2014
   CCF-1217462.  % Started June 2012
   A preliminary version of the paper appeared in the 31st
   International Symposium on Computational Geometry (SoCG 2015)
   \cite{chr-fpuvp-15}.%
}

\DCGVer{%
   \bibliographystyle{spmpsci}%
}
\NotDCGVer{
   \bibliographystyle{salpha}%
}

\bibliography{beer}%

\appendix

%%%%%%%%%%%%%%%%%%%%%%%%%%%%%%%%%%%%%%%%%%%%%%%%%%%%%%%%%%%%%%%%% 
%%%%%%%%%%%%%%%%%%%%%%%%%%%%%%%%%%%%%%%%%%%%%%%%%%%%%%%%%%%%%%%%% 

\section{An Integral Calculation}
\apndlab{h:p:minima}%

% \begin{figure}[t]\centering
%     \includegraphics[width=.2\linewidth]{figs/integral}%
%    %
%     \caption{Graph of the function $xy=c$}
%     \figlab{F:grid}
% \end{figure}

\begin{lemma}
    \lemlab{integral}%
    Let $\VolUArea{d}{\Delta}$ be the total measure of the points
    $\pnt = (\pnt_1,\ldots, \pnt_d)$ in the hypercube $[0,1]^d$, such
    that $\pv\pth{\pnt} = \pnt_1 \pnt_2 \cdots \pnt_d \leq \Delta$.
    That is, $\VolUArea{d}{\Delta}$ is the measure of all points in
    hypercube with point volume at most $\Delta$.  Then
    \[
        \Big.%
        \VolUArea{d}{\Delta} = \sum_{i=0}^{d-1} \frac{\Delta}{i!} {
           \ln^{i} \frac{1}{\Delta} }.
    \]
\end{lemma}
\begin{proof}
    The claim follows by tedious but relatively standard
    calculations. As such, the proof is included for the sake of
    completeness.

    \parpic[r]{%
       \includegraphics[width=.2\linewidth]{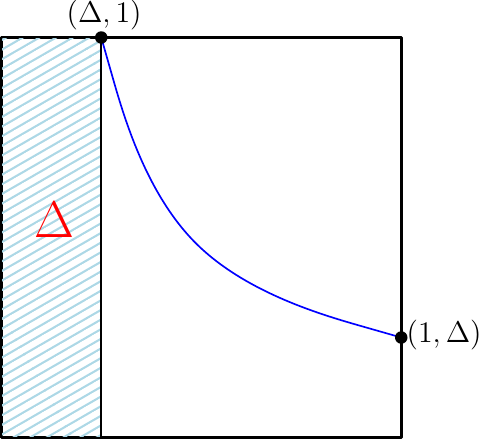}%
    }

    \InHsienRadical{\picskip{6}\linenumbers}
    \noindent The case for $d=1$ is trivial.  Consider the $d=2$ case.
    Here the points whose point volume equals $\Delta$ are defined
    by the curve $xy = \Delta$.  This curve intersects the unit square
    at the point $(\Delta,1)$.  As $\VolUArea{d}{\Delta}$ is the total
    volume under this curve in the unit square we have that
    \begin{align*}
        \VolUArea{2}{\Delta}%
        =%
        \Delta + \int_{x=\Delta}^1 \frac{\Delta}{x} \dInt x%
        =% 
        \Delta + \Delta \ln \frac{1}{\Delta}.
    \end{align*}
    In general, we have
    \begin{align*}
        \frac{1}{(d-1)!} \int_{x=\Delta}^1 \frac{\Delta}{x} \ln^{d-1}
        \frac{x}{\Delta} \dInt x%
        =%
        \frac{\Delta}{(d-1)!} \pbrcx{ \frac{1}{d} \ln^{d}
           \frac{x}{\Delta} }_{x=\Delta}^1 %
        =%
        \frac{\Delta}{d!}  \ln^{d} \frac{1}{\Delta}. %
    \end{align*}
    Now assume inductively that
    \begin{align*}
        \VolUArea{d-1}{\Delta}%
        =%
        \sum_{i=0}^{d-2} \frac{1}{i!} \Delta \ln^{i} \frac{1}{\Delta},
    \end{align*}
    then we have
    \begin{align*}
      \VolUArea{d}{\Delta}%
      &=%
        \Delta + \int_{x_d=\Delta}^1
        \VolUArea{d-1}{\frac{\Delta}{x_d}} \dInt x_d%
        =%
        \Delta + \int_{x_d=\Delta}^1 \pth{ \sum_{i=0}^{d-2} \frac{
        \Delta}{i! x_d} \ln^{i} \frac{x_d}{\Delta} }
        \dInt x_d\\
      &=%
        \Delta + \sum_{i=0}^{d-2} \frac{1}{i!} \pth{
        \int_{x_d=\Delta}^1 \frac{\Delta}{x_d} \ln^{i}
        \frac{x_d}{\Delta} \dInt x_d}%
        =%
        \Delta + \sum_{i=1}^{d-1} \frac{\Delta}{i!} { \ln^{i}
        \frac{1}{\Delta} } =%
        \sum_{i=0}^{d-1} \frac{\Delta}{i!} { \ln^{i} \frac{1}{\Delta}
        }. 
        \DCGVer{{\tag*\qed}}
    \end{align*}
\end{proof}%

\end{document}